\documentclass[sigconf]{aamas} 
\usepackage{balance}
\usepackage{xparse}
\usepackage{xargs}

\usepackage[
    disable, %
    textsize=tiny,
]{todonotes}

\newcommandx{\bleg}[2][1=]{\todo[linecolor=blue,backgroundcolor=blue!25,bordercolor=blue,#1]{Q: #2}}
\newcommand{\citeps}{\cite}
\newcommand{\citets}{\cite}

\usepackage{amsmath}
\usepackage{amsfonts}
\usepackage{amsthm}
\usepackage{mathtools}
\usepackage{nicefrac}

\usepackage{graphicx} %
\usepackage{verbatim}
\usepackage{algpseudocode} %
\usepackage{algorithm}
\usepackage{pgfplots}
\pgfplotsset{compat=1.17}

\usepackage{hyperref}

\usepackage{enumitem}
\usepackage{microtype}
\usepackage{cleveref}

\usepackage{changepage}

\usepackage{subcaption}
\usepackage{booktabs} %
\usepackage{tikz}
\usetikzlibrary{calc}

\setcopyright{rightsretained}
 \acmConference[GAIW'26]{Appears at the 8th Games, Agents, and Incentives Workshop (GAIW-26). Held as part of the Workshops at the 25th International Conference on Autonomous Agents and Multiagent Systems.}{May 2026}{Paphos, Cyprus}{Armstrong, Curry, Hosseini, Mattei, Tsang, Wąs (Chairs)} 
\acmISBN{}

\acmSubmissionID{40}

\NewDocumentCommand{\Reals}{}{\mathbb{R}}
\NewDocumentCommand{\bools}{}{\{0,1\}}%

\DeclareMathOperator*{\argmax}{arg\,max}

\DeclareMathOperator*{\finset}{finset}

\NewDocumentCommand{\Prob}{}{\mathbf{P}}
\NewDocumentCommand{\Probof}{m}{\Prob\left[#1\right]}
\NewDocumentCommand{\Expect}{}{\mathbf{E}}
\NewDocumentCommand{\Expectof}{m}{\Expect\left[#1\right]}
\NewDocumentCommand{\Expectofu}{o m}{\Expect_{#1}\left[#2\right]}
\NewDocumentCommand{\Powerset}{m}{\mathcal{P}\left(#1\right)}

\NewDocumentCommand{\Actions}{}{\mathcal{X}}
\NewDocumentCommand{\Samples}{}{\Omega}
\NewDocumentCommand{\Sigalg}{}{\mathcal{F}}

\NewDocumentCommand{\infogood}{}{\mathbf{I}}
\NewDocumentCommand{\zero}{}{\mathbf{0}}

\NewDocumentCommand{\BuyerContext}{}{\textsc{BuyerContext}}
\NewDocumentCommand{\InfoOffer}{}{\textsc{InfoOffer}}
\NewDocumentCommand{\InfoOffers}{}{\textsc{InfoOffers}}
\NewDocumentCommand{\BuyerContexts}{}{\textsc{BuyerContexts}}
\NewDocumentCommand{\Seller}{}{\textsc{Seller}}
\NewDocumentCommand{\Sellers}{}{\textsc{Sellers}}
\NewDocumentCommand{\RIP}{}{\textsc{RIP}}
\NewDocumentCommand{\LLM}{}{\texttt{LLM}}
\NewDocumentCommand{\prompt}{}{\texttt{prompt}}

\NewDocumentCommand{\infooffers}{}{\mathcal{I}}

\NewDocumentCommand{\infonomyurl}{}{\url{https://anonymous.4open.science/r/infonomy-server-5668/}}
\NewDocumentCommand{\infonomy}{}{\texttt{infonomy-server}}

\theoremstyle{plain}
\newtheorem{theorem}{Theorem}[section]
\newtheorem{nontheorem}[theorem]{Non-theorem}

\newtheorem{lemma}[theorem]{Lemma}

\theoremstyle{definition}
\newtheorem{definition}[theorem]{Definition}
\newtheorem{example}[theorem]{Example}

\theoremstyle{remark}

\newtheorem*{notation*}{Notation}

\newcommand{\Struct}[1]{\State \textbf{class} \textsc{#1}}
\newcommand{\EndStruct}{\State \textbf{end class}}
\newcommand{\LeftComment}[1]{\State \(\triangleright\) #1}

\title[Recursive Information Markets]{Extrapolating Volition with Recursive Information Markets}

\author{Abhimanyu Pallavi Sudhir}
\affiliation{
  \institution{University of Warwick}
  \city{Coventry}
  \country{United Kingdom}}
\email{abhimanyu.pallavi-sudhir@warwick.ac.uk}

\author{Long Tran-Thanh}
\affiliation{
  \institution{University of Warwick}
  \city{Coventry}
  \country{United Kingdom}}
\email{long.tran-thanh@warwick.ac.uk}


\begin{abstract}  
 
\noindent\textbf{Background:} A key challenge in information economics and AI alignment is that of efficiently valuing or scoring information supplied by a seller or language model that is potentially more-informed than the buyer or evaluator. This is known as the problem of ``information asymmetry'' in economics or ``scalable oversight'' in AI alignment.

\noindent\textbf{Objectives and Research Questions:} We ask how to formalize value-of-information under recursive inspection, whether deeper inspection can be made decision-theoretically principled, and how such mechanisms can support scalable oversight beyond standard RLHF.

 \noindent\textbf{Methods:} We introduce a Bayesian framework for recursive information valuation, compare a naive successive protocol to a recursive protocol modeled as an imperfect-recall game, prove ex-ante optimality against admissible protocols, and analyze a marginal-value reward mechanism for scalable oversight within our Bayesian framework.

\noindent\textbf{Results:} We show ex-post inspection alone can still disincentivize corrective context, provide a counterexample to naive recursion, prove the Recursive Inspection Protocol is ex-ante superior to any admissible purchase protocol, characterize equilibrium behavior for the marginal-value mechanism, and present a working server implementation (\texttt{infonomy-server}).

\noindent\textbf{Conclusions:} Recursive inspection offers a principled way to price information under persistent asymmetry and a practical path for market-based oversight; however, our current scalable-oversight mechanism remains imperfect, motivating tighter future guarantees on equilibrium shortfall.

\end{abstract}

\begin{document}

\maketitle 

\section{Introduction}

A key challenge in both economics and machine learning is the development of mechanisms for efficiently pricing information \citeps{stigler1961economics}. In settings where ground truth is available (e.g. supervised learning, or prediction markets for well-defined events), information may be priced with proper scoring rules \citep{hansonLogarithmicMarketScoring2002}; when ground truth is not available, one must rely on human buyers (an information market) or evaluators (e.g. reinforcement learning via human feedback/RLHF) to price it.

The core obstacle in valuing information based on subjective preferences is \emph{information asymmetry}: the seller (or information-giver) by definition possesses information the buyer (or evaluator) doesn't, which leads to a ``Market of Lemons'' as famously described in \citeps{lemons}, so the prices given by the buyer only reflect her superficial preferences (based on the information she has) rather than what her true preferences would be with full information. An analogous problem arises in AI alignment: techniques such as RLHF fundamentally rely on a human's ability to evaluate the outputs of increasingly capable and eventually superhuman AI models \citeps{burnsWeaktoStrongGeneralizationEliciting2024,caspar2023openrlhf,sudhir2025ScalableOversightBenchmark}---this is known as the problem of \emph{scalable oversight}.

Recently, \citets{weissRedesigningInformationMarkets2024} proposed the \emph{Information Bazaar}: an information market mechanism that mitigates information asymmetry \textbf{by using Large Language Model (LLM) agents to make purchase decisions}. Specifically, the mechanism addresses the \emph{buyer's inspection paradox} \citep{arrow1972economic,vanAlstyne1999AproposalValuingInformationInstrumentalGoods}: the problem that, unlike with buying other goods, someone buying information definitionally does not \emph{know} what the information she is about to buy is\footnote{The seller could of course reveal part of the information, e.g. ``metadata'' to advertise the information---which creates a trade-off between information asymmetry and having positive externalities}. Their mechanism lets the buyer use an LLM agent to ``inspect'' the information and make the purchase decision with full knowledge of the information.

In this work, we introduce a formal \textbf{Bayesian framework for analyzing mechanisms to score information under information asymmetry}, and use this to study both market mechanisms (like in \citets{weissRedesigningInformationMarkets2024}) and scoring rules that can be used to train AI models (such as Language Models) to supply more valuable information.

Our findings and contributions are as follows:

\begin{enumerate}
  \item \textbf{Recursive inspection protocol.} We observe that the Information Bazaar mechanism in \citets{weissRedesigningInformationMarkets2024} does not eliminate information asymmetry, because the LLM buyer inspecting information can still lack other pieces of information that are known to the buyer and correlated with the one it is purchasing. We build on their work, and discover that the straightforward method of ``simply applying \citets{weissRedesigningInformationMarkets2024} to itself'' is too simplistic---and instead introduce a more robust protocol we call the \emph{Recursive Inspection Protocol}, formulated as an imperfect-recall game.
  \item \textbf{Scalable oversight mechanism.} We express the scalable oversight problem in our Bayesian framework, and construct an example scalable oversight mechanism that generalizes the ``AI safety via market-making'' proposal \citeps{evanhubingerAISafetyMarket2020} to problems beyond binary forecasting.
  \item \textbf{Practical implementation.} We provide an implementation of an information market server implementing the Recursive Inspection Protocol, detailed in \Cref{sec:practical}, which can directly be applied to various practical applications for information markets such as question-and-answer sites, product inspections and online fact-checking.
\end{enumerate}

\subsection{Related work}
\label{sec:related}

\paragraph{Information economics.} The foundational theory of information economics\bleg{feel like this might elicit some annoying comments from reviewers} is the \emph{value-of information} framework \citeps{kuhn1953information, howard1966information, raiffa1961AppliedStatisticalDecisionTheory}, which treats information as an instrumental good \citeps{vanalstyneProposalValuingInformation1999}. The buyer's inspection paradox was introduced by \citets{arrow1972economic} and named by \citets{vanalstyneProposalValuingInformation1999}. \citets{hansonIPBarbedWire2011,hansonRahEfficientIP2011} discussed the inefficiency of information markets in the context of intellectual property (IP) law, commenting: ``\emph{just as farmers developed barbed-wire, someday I expect IP advocates will develop better forms of intellectual property}''.

\paragraph{Mechanism design for information markets.} Simple, naive information markets suffer from a number of flaws: \emph{the low cost of duplication} \citeps{samuelson2009economics}, \emph{the transaction costs of tenders} \citeps{tenders}, the transaction cost of \emph{learning} new information, and \emph{information asymmetry} \citeps{arrow1972economic,vanalstyneProposalValuingInformation1999}. \citets{conitzerPredictionMarketsMechanism2012} introduced a mechanism for rewarding information-providing agents based on their influence on prediction market prices, though this is only applicable in contexts where ground-truth (prediction market resolutions) is available.

\paragraph{Scalable oversight} The fundamental limitation of RLHF that it relies on a human's ability to judge a (potentially superhuman) AI's outputs has long been recognized \citeps{christiano2018IDA}, and is known as the \emph{scalable oversight} problem in the AI alignment literature \citeps{hubinger2020alignmentProposalsComplexityClasses,bowmanMeasuringProgressScalable2022}. One well-known scalable oversight proposal is Debate \citeps{irvingAISafetyDebate2018}; our proposal to augment RLHF with information markets may be seen as yet another such proposal.

\paragraph{Mechanism design with LLMs.} Market design for LLM participants has opened up many new frontiers previously not possible with humans alone. Apart from the information bazaar \citeps{weissRedesigningInformationMarkets2024}, this includes e.g. token auctions for online advertising \citeps{duttingMechanismDesignLarge2024}, economic simulations with AI agents \citeps{liEconAgentLargeLanguage2024a}, and forecasting with LLMs \citeps{schoeneggerWisdomSiliconCrowd2024,halawiApproachingHumanLevelForecasting2024,palekaConsistencyChecksLanguage2024,buterin2024InfoFinance}.

\paragraph{Zero-knowledge proofs.} Another way for a seller to prove the value of his information without revealing it is via a \emph{zero-knowledge proof} \citeps{goldreich2001FoundationsCryptographyV1} -- in formal settings, this is available if the information is a solution to a PSPACE problem \citeps{zkPSPACE1,zkPSPACE2}; extending this to informal settings is an active area of work \citeps{hammondNeuralInteractiveProofs2024}.

\paragraph{Miscellaneous} Recent works like \citets{manelli2006Bundling,babiaoff2012OptimalMechanismsSellingInformation,bergemannDesignPriceInformation2018} have studied ``optimal mechanisms for selling information'', but from the point-of-view of a seller maximizing his revenue -- we, on the other hand, are interested in improving the efficiency of the information market itself. Information market mechanisms have also been designed for \emph{data markets} in machine learning, e.g. \citets{fallah2024ThreeLayerDataMarkets,chenSellingDataMachine2022,ghorbani2019DataShapley}.\bleg{all papers in this para are totally irrelevant but I suspect reviewers may bring them up}

\section{Bayesian setting}
\label{sec:bayes}

We are concerned with modeling an agent $\alpha$'s ``value for information'' in an expected utility maximization framework. We assume a probability space $(\Samples,\Sigalg,\Prob)$ (with $\Prob$ a common prior for all agents ever discussed) and that the only source of utility is some decision problem specified by a set of choices $\Actions$ for $\alpha$ and payoffs given by a measurable utility function $U:\Samples\times\Actions\to\Reals$. An ``information good'' is a tuple of a random variable $I$, its realization or ``true value''\footnote{Including $i$ in the tuple is to simplify notation when talking about taking conditional expectations on $I$. It is not necessary: instead of writing $\Expectof{U(x)\mid I=i}$ we can think of $\Expectof{U(x)\mid I}$ as itself a random variable correlated with $I$} $I(\omega)=i$ and a price: $\infogood=\langle I, i, p\rangle$. All of these contents of an information good may be hidden from $\alpha$. We want to study how $\alpha$ values informational goods.

With no further information, $\alpha$'s choice would maximize its utility over its prior, i.e. choose $\argmax_{x\in\Actions}\Expectof{U(x)}$. With information $\langle I, i, p\rangle$, the agent would instead maximize its utility over its posterior, i.e. $\argmax\Expectof{U(x)\mid I=i}$. Thus we can say the utility of that information is:

\begin{multline}
  U^1(\infogood)=U(\argmax\Expectof{U(x)\mid I=i}) \\
  - U(\argmax\Expectof{U(x)}) - p
  \label{eq:voi-true}
\end{multline}

If $\alpha$ has to decide whether to buy an information good (or decide which one to buy out of a list of offers), it has to take the expectation of $U^1(\infogood)$. There are several ways to do this. There is the \emph{ex-post} value of information $\Expectof{U^1(\infogood)|I=i}$, which is $\alpha$'s estimate for $\infogood=\langle I,i,p\rangle$ after viewing it:

\begin{multline}
\Expectof{U^1(\infogood)|I=i}=\max_{x\in\Actions}\Expectof{U(x)\mid I=i} \\
- \Expectof{U\left(\argmax_{x\in\Actions}\Expectof{U(x)}\right)\mid I=i} - p
\label{eq:voi-ex-post}
\end{multline}  

The \emph{ex-ante} value (before seeing the information) may be taken as the expectation over \Cref{eq:voi-ex-post}: $\Expectofu[i\sim I]{\Expectof{U^1(\infogood)|I=i}}$ (which would be the \emph{value of an experiment} \citeps{lindley1956OnMeasureInformationExperiment}), though if $\alpha$ is not even aware in advance which random variable $I$ will be revealed by information good $\infogood$, then we would further need to assume a prior over the process generating $\infogood$ and take the expetation over it, i.e. $\Expectofu[\infogood\sim \Probof{\infogood}]{\Expectofu[i\sim I]{\Expectof{U^1(I\mid I=i)}}}$. In the absence of any inspection of the information before purchasing/scoring it, that ex-ante value would be our value for $\infogood$.

In the case where $\alpha$ is an evaluator providing human feedback to an AI model, the evaluator can see the information before scoring it. The same is the case in the Information Bazaar of \citets{weissRedesigningInformationMarkets2024}, where an LLM agent inspects the information before deciding whether to buy it for its principal. In these settings, $\alpha$ will value the information at its ex-post value $\Expectof{U^1(\infogood)|I=i}$.

However, ex-post VOI is still not enough: while knowing $I=i$ (inspecting $\infogood$) provides \emph{some} information on $U(\infogood)$, it does not provide all of it; there may still be information asymmetry, because $U(\infogood)$ is itself a random variable not completely determined by $\infogood$ (much like in ordinary goods markets where you can inspect the good but still have information asymmetry). The following example demonstrates this.

\begin{example}

Consider a forecasting decision problem where the agent reports a probability $x\in[0,1]$ for an event $E$, with log-score utility
  \[
  U(x)=
  \begin{cases}
  \log x, & \text{if } E \text{ occurs}\\
  \log(1-x), & \text{otherwise.}
  \end{cases}
  \]
  Suppose $\Prob(E)=0.1$, and there are two random variables $I_1,I_2$ such that
  \[
  \Prob(E\mid I_1=1)=0.4,\qquad \Prob(E\mid I_1=1,I_2=1)=0.2.
  \]
  One concrete joint distribution with these properties is shown in \Cref{tab:factcheck-example}.

  \begin{table}[h]
  \centering
  \caption{A distribution exhibiting the fact-checking effect.}
  \label{tab:factcheck-example}
  \begin{tabular}{cccc}
  \toprule
  $\Prob(I_1,I_2)$ & $I_1$ & $I_2$ & $\Prob(E\mid I_1,I_2)$ \\
  \midrule
  $15/32$ & $0$ & $0$ & $0.08$ \\
  $15/32$ & $0$ & $1$ & $0.08$ \\
  $1/24$  & $1$ & $0$ & $0.5$ \\
  $1/48$  & $1$ & $1$ & $0.2$ \\
  \bottomrule
  \end{tabular}
  \end{table}

  Now suppose an information seller knows that $(I_1,I_2)=(1,1)$. If the seller reveals only $I_1=1$, the buyer updates from $0.1$ to $0.4$, so the ex-post gain is $\log(0.4)-\log(0.1)$. If the seller reveals both $(I_1,I_2)=(1,1)$, the buyer updates from $0.1$ to $0.2$, yielding only $\log(0.2)-\log(0.1)$. Hence the seller is incentivized to reveal only $I_1$.

  This illustrates a fact-checking failure mode: $I_1$ can be interpreted as a persuasive claim and $I_2$ as additional context that weakens that claim. Under a mechanism that rewards only immediate ex-post value, providing the corrective context is disincentivized.

  (\emph{Sidenote}: Our presentation of ``false'' claims may be a bit confusing. Since we're dealing in a Bayesian setting, we do not suppose that the AI/information-seller can directly give a false value for a random variable---rather, the random variable $I_1$ may be interpreted as ``what the AI says about some underlying (not directly observed) random variable $J_1$'' etc.)
\end{example}

We instead provide two mechanisms: one, (1) Recursive Information Protocol, for valuing information in markets with information asymmetry (\cref{sec:market}) and (2) a scalable oversight or scoring mechanism to supply ``more fully informed'' human feedback to AI models during training (\cref{sec:scoring}). 

\section{Market mechanisms}
\label{sec:market}

One way to think of this persistence of information asymmetry is: \Cref{eq:voi-true} creates a \emph{new} decision problem for $\alpha$, that of deciding whether to buy $\infogood$ -- or which information to buy out of a list of offers. The set of information goods offered $\infooffers=\{\infogood_1,\dots\infogood_k\}$ is a new decision problem, with utilities of each choice now given by the (unknown/random, much like $U(x)$) true value-of-information $U(\infogood)$. Thus we may be offered another set of information goods $\{\infogood_1^1,\dots\infogood_k^1\}$ to help us with this decision, ad recursum. In general we have a sequence of decision problems $\Actions^{n+1}=\Powerset{\infooffers^n}$ where $\infooffers^n:=\{\infogood_1^n,\dots\infogood_k^n\}$ are the information goods offered to us to help us decide $\Actions^n$, and each choice corresponds to choosing a subset of that information to help decide $\Actions^n$, where $\Actions^0 =\Actions$ and $\Actions^1=\Powerset{\infooffers^0}=\Powerset{\infooffers}$.

There are two ways that we can set up these recursive problems. The first we call the \emph{successive inspection protocol}, which we describe in \Cref{app:naive}. However, this approach, while conceptually simpler, has its limitations---and in \Cref{app:nonnaive} we will instead introduce the superior \emph{recursive inspection protocol}.

\subsection{Successive Inspection Protocol}
\label{app:naive}

The successive inspection protocol protocol arises from ``simply applying \citet{weissRedesigningInformationMarkets2024} to itself'': i.e. we model each decision problem as having its own utility function $U^n:\Actions^n\to\Reals$ based on its instrumental utility for the previous decision problem $\Actions^{n-1}$.

\begin{table}[h]
\centering
\caption{\emph{Successive decision problems, naive approach}}
\vspace{0.2em}
\label{tab:sdp}
\begin{tabular}{lll}
\vspace{0.2em}
\textbf{Choice set} & \textbf{True utilities} & \textbf{Information Offers} \\

\vspace{0.2em}
$\Actions$ & $U:\Actions\to\Reals$ & $\infooffers=\{\zero,\infogood_1,\infogood_2,\dots\}$ \\

\vspace{0.2em}
$\Actions^1=\infooffers$ & $U^1:\Actions^1\to\Reals$ & $\infooffers^1=\{\zero,\infogood^1_1,\infogood^1_2,\dots\}$ \\

\vspace{0.2em}
$\Actions^2=\infooffers^1$ & $U^2:\Actions^2\to\Reals$ & $\infooffers^2=\{\zero,\infogood^2_1,\infogood^2_2\dots\}$ \\
$\dots$ & $\dots$ & $\dots$ \\
\end{tabular}
\end{table}

\begin{multline}
U^{n+1}\left(\langle I, i, p\rangle\right) := U^n\left(\argmax_{x\in\Actions^n}\Expectof{U^n(x)\mid I=i}\right) \\
- U^n\left(\argmax_{x\in\Actions^n}\Expectof{U^n(x)}\right) - p
\label{eq:voi-rim}
\end{multline}

These choice sets and utilities are very similar to the decision problems created by the \emph{recursive inspection protocol} described in \Cref{alg:rim} and in the main body; except that each action $x^n\in\Actions^n$ is made only consulting the information chosen in $x^{n+1}\in\Actions^{n+1}$. We can then say that the decisions made under this protocol are:

\begin{equation}
x^n_*=\argmax_{x\in\Actions^n} \Expectof{
  U^n(x)\mid x^{n+1}_*
}
\label{eq:rim-choices}
\end{equation}

In general this is an ill-defined infinite recursion. But finite restrictions of this are natural, stopping at some fixed $x^N_{*:N}=\argmax_{x\in\Actions_N}\Expectof{U^N(x)}$ (you can think of this stopping as caused by the transaction costs of inspection), so that for $n<N$:

\begin{equation}
x^n_{*:N}=\argmax_{x\in\Actions^n} \Expectof{
  U^n(x)\mid x^{n+1}_{*:N}
}
\label{eq:rim-choices-fin}
\end{equation}

While this approach is suitable for settings where all sellers have identical information (e.g. in an AI alignment setting), is it fails to account for the possibility that a choice $x^n$ can directly (i.e. not through its impact on $x^{n-1}$) impact a choice $x^m$ where $m<n-1$, as demonstrated in the following example.

\begin{proof}[Counter-example]
Consider the decision problem with action set $\Actions^0=\{x_0, x_1, x_2\}$. We interpret $x_0$ as ``eat raw legume'', $x_1$ as ``eat rice'' and $x_2$ as ``eat boiled legume''. In reality we have $U(x_2)>U(x_1)>>U(x_0)$ (rice is unhealthy, but raw legumes are toxic). 

Suppose the first-level information offers are $\infooffers^0=\{I^1_0,I^1_1\}$, where $I^1_0$ states ``legumes are toxic'' and $I^1_1$ states ``rice is unhealthy''. Suppose the second-level information offers are $\infooffers^1=\{I^2_0\}$, where $I^2_0$ states ``the toxins in legumes can be removed by boiling''. All of these information offers could even be free.

Then the optimal action is $(x_2,\{I^1_0,I^1_1\},\{I^2_0\})$; however, if the information bought in level-2 $\{I^2_0\}$ is not available while deciding $x^0$, the best the agent can do is $(x_1,\{I^1_0\},\{I^2_0\})$ to prevent itself from eating raw legumes (since it will not know that the toxins can be removed by boiling).
\end{proof}

\subsection{Recursive Inspection Protocol}
\label{app:nonnaive}

Instead our approach, the \emph{recursive information protocol} implemented in \Cref{sec:practical}, allows the agent (or rather the LLM subcontracted by the agent) to retain the full sequence of information bought in the recursive steps $x^{n+1}_*,\dots x^N$ while making the decision $x^n\in\Actions^n$ (where $N$ is some pre-defined finite depth we recurse till); furthermore $x^n$ is decided keeping in mind the full traceback of decision problems $\Actions^0,\dots\Actions^{n-1}$ that may be influenced by this decision. This is naturally modelled as an \emph{imperfect recall game}\footnote{The decision-theoretic considerations in imperfect recall games do not matter to us, as the agent's actions themselves are not being forgotten --- only the information offers} \citeps{imperfectRecall2024} where we first decide $x_*^N\in\Actions^N$ with full information $\infooffers^0\cup\dots\infooffers^{N-1}$, then $x_*^{N-1}\in\Actions^{N-1}$ with information $\infooffers^0\cup\dots\infooffers^{N-2}\cup x_*^{N}$ and so on until we finally decide $x_*^0\in\Actions^0$ with information $x_*^1\cup\dots \cup x_*^N$. This is shown in \Cref{fig:imperfect}.

A node $(x^n,\dots x^N)$ corresponds to the state where the agent has purchased $(x^n,\dots x^N)$ and is choosing some $x^{n-1}\in\Actions^{n-1}$. For a Bayesian agent, we can thus recursively give the ``value of being at a node'':

\begin{multline}
  U(x^n,\dots x^N) = U\Big(
    \argmax_{x\in\Actions^{n-1}}
    \Expect\bigl[
      U(x,x^n,\dots x^N) \mid \\
      \infooffers^0\cup\dots\cup\infooffers^{n-2}\cup
      x^n\cup\dots\cup x^N
    \bigr],
    x^n,\dots x^N
  \Big)
  \label{eq:imperfect-u}
\end{multline}
\begin{equation}
  U(x^0,\dots x^N) = U(x^0) - \sum_{n=1}^N \sum_{\substack{\langle I, i, p\rangle \\ \in x^n}} p
  \label{eq:imperfect-u0}
\end{equation}

This then completely specifies the behavior of a Bayesian agent performing a depth-$N$ recursive inspection:

\begin{multline}
  x^n_*=\argmax_{x\in\Actions^n}\Expect\bigl[
      U(x,x^{n+1},\dots x^N) \mid
      \infooffers^0\cup\dots \\
      \cup\infooffers^{n-1}\cup
      x^{n+1}\cup\dots\cup x^N
  \bigr]
  \label{eq:imperfect-x}
\end{multline}

\begin{figure}
  \centering
  \resizebox{\columnwidth}{!}{
  \def\levely{106}              %
\def\chosenx{-38}             %
\def\otherxstart{33}          %
\def\otherxskip{49}           %
\def\labelside{left}          %
\def\labeldist{4pt}          %
\def\arrowlabelpos{0.5}       %
\def\specialarrowpos{0.8}     %
\def\lastlayerfrac{0.6}       %
\def\chosencolor{black}       %
\def\othercolor{gray!75}      %
\def\coverrectlayer{3}        %

\tikzset{every picture/.style={line width=0.75pt}}

\begin{tikzpicture}[x=0.75pt,y=0.75pt,yscale=-1,xscale=1]

\coordinate (L0) at (331.71, 58.43);
\fill[\chosencolor] (L0) circle (4.71pt);
\node[align=left,\labelside=\labeldist] at (L0) {$\mathcal{I}^{0} \cup \dotsc \cup \mathcal{I}^{N-1}$};

\coordinate (L1C) at ($(L0) + (\chosenx, \levely)$);
\coordinate (L1O1) at ($(L0) + (\otherxstart, \levely)$);
\coordinate (L1O2) at ($(L0) + (\otherxstart + \otherxskip, \levely)$);
\coordinate (L1O3) at ($(L0) + (\otherxstart + 2*\otherxskip, \levely)$);

\fill[\chosencolor] (L1C) circle (4.71pt);
\fill[\othercolor] (L1O1) circle (4.71pt);
\fill[\othercolor] (L1O2) circle (4.71pt);
\fill[\othercolor] (L1O3) circle (4.71pt);

\draw[->, >=latex, line width=0.75pt, \chosencolor] (L0) -- (L1C) 
    node[pos=\arrowlabelpos, left] {$x_{*}^{N}$};
\draw[->, >=latex, line width=0.75pt, \othercolor] (L0) -- (L1O1) 
    node[pos=\arrowlabelpos, right, \othercolor] {$x_{1}^{N}$};
\draw[->, >=latex, line width=0.75pt, \othercolor] (L0) -- (L1O2) 
    node[pos=\arrowlabelpos, right, \othercolor] {$x_{2}^{N}$};
\draw[->, >=latex, line width=0.75pt, \othercolor] (L0) -- (L1O3) 
    node[pos=\arrowlabelpos, right, \othercolor] {$x_{3}^{N}$};
\node[right=3pt, \othercolor] at ($(L1O3) + (10pt, 0)$) {$\dotsc$};

\node[align=left, \labelside=\labeldist] at (L1C) 
    {$\begin{array}{l}\mathcal{I}^{0} \cup \dotsc \cup \mathcal{I}^{N-2}\\\cup x_{*}^{N}\end{array}$};

\coordinate (L2C) at ($(L1C) + (\chosenx, \levely)$);
\coordinate (L2O1) at ($(L1C) + (\otherxstart, \levely)$);
\coordinate (L2O2) at ($(L1C) + (\otherxstart + \otherxskip, \levely)$);
\coordinate (L2O3) at ($(L1C) + (\otherxstart + 2*\otherxskip, \levely)$);

\fill[\chosencolor] (L2C) circle (4.71pt);
\fill[\othercolor] (L2O1) circle (4.71pt);
\fill[\othercolor] (L2O2) circle (4.71pt);
\fill[\othercolor] (L2O3) circle (4.71pt);

\draw[->, >=latex, line width=0.75pt, \chosencolor] (L1C) -- (L2C) 
    node[pos=\arrowlabelpos, left] {$x_{*}^{N-1}$};
\draw[->, >=latex, line width=0.75pt, \othercolor] (L1C) -- (L2O1) 
    node[pos=\arrowlabelpos, right, \othercolor] {$x_{1}^{N-1}$};
\draw[->, >=latex, line width=0.75pt, \othercolor] (L1C) -- (L2O2) 
    node[pos=\arrowlabelpos, right, \othercolor] {$x_{2}^{N-1}$}; %
\draw[->, >=latex, line width=0.75pt, \othercolor] (L1C) -- (L2O3) 
    node[pos=\arrowlabelpos, right, \othercolor] {$x_{3}^{N-1}$}; %
\node[right=3pt, \othercolor] at ($(L2O3) + (10pt, 0)$) {$\dotsc$};

\node[align=left, \labelside=\labeldist] at (L2C) 
    {$\begin{array}{l}\mathcal{I}^{0} \cup \dotsc \cup \mathcal{I}^{N-3}\\\cup x_{*}^{N-1}\cup x_{*}^{N}\end{array}$};

\coordinate (L3C) at ($(L2C) + (\chosenx, \levely)$); %
\coordinate (L3O1) at ($(L2C) + (\otherxstart, \levely)$);
\coordinate (L3O2) at ($(L2C) + (\otherxstart + \otherxskip, \levely)$);
\coordinate (L3O3) at ($(L2C) + (\otherxstart + 2*\otherxskip, \levely)$);

\fill[\chosencolor] (L3C) circle (4.71pt);
\fill[\othercolor] (L3O1) circle (4.71pt);
\fill[\othercolor] (L3O2) circle (4.71pt);
\fill[\othercolor] (L3O3) circle (4.71pt);

\draw[->, >=latex, line width=0.75pt, \chosencolor] (L2C) -- (L3C) 
    node[pos=\specialarrowpos, left] {$x_{*}^{1}$};
\draw[->, >=latex, line width=0.75pt, \othercolor] (L2C) -- (L3O1) 
    node[pos=\specialarrowpos, right, \othercolor] {$x_{1}^{1}$}; %
\draw[->, >=latex, line width=0.75pt, \othercolor] (L2C) -- (L3O2) 
    node[pos=\specialarrowpos, right, \othercolor] {$x_{2}^{1}$}; %
\draw[->, >=latex, line width=0.75pt, \othercolor] (L2C) -- (L3O3) 
    node[pos=\specialarrowpos, right, \othercolor] {$x_{3}^{1}$}; %
\node[right=3pt, \othercolor] at ($(L3O3) + (10pt, 0)$) {$\dotsc$};

\node[align=left, \labelside=\labeldist] at (L3C) 
    {$x_{*}^{1}\cup \dotsc \cup x_{*}^{N}$};

\coordinate (L2_rect_start) at ($(L2C)!.2!(L3C)$);
\coordinate (L2_rect_end) at ($(L2C)!.65!(L3C)$);
\filldraw[fill=white, draw=none] ($(L2_rect_start) - (110, 0)$) rectangle ($(L2_rect_end) + (110, 0)$);

\node[\chosencolor] at ($(L2C) + (0.425*\chosenx, 0.425*\levely)$) {$\dotsc$};

\coordinate (L4C) at ($(L3C) + (\chosenx, \levely * \lastlayerfrac)$);
\coordinate (L4O1) at ($(L3C) + (\otherxstart, \levely * \lastlayerfrac)$);
\coordinate (L4O2) at ($(L3C) + (\otherxstart + \otherxskip, \levely * \lastlayerfrac)$);
\coordinate (L4O3) at ($(L3C) + (\otherxstart + 2*\otherxskip, \levely * \lastlayerfrac)$);

\fill[\chosencolor] (L4C) circle (4.71pt);
\fill[\othercolor] (L4O1) circle (4.71pt);
\fill[\othercolor] (L4O2) circle (4.71pt);
\fill[\othercolor] (L4O3) circle (4.71pt);

\draw[->, >=latex, line width=0.75pt, \chosencolor] (L3C) -- (L4C) 
    node[pos=\arrowlabelpos, left] {$x_{*}^{0}$};
\draw[->, >=latex, line width=0.75pt, \othercolor] (L3C) -- (L4O1) 
    node[pos=\arrowlabelpos, left, \othercolor] {$x_{1}^{0}$};
\draw[->, >=latex, line width=0.75pt, \othercolor] (L3C) -- (L4O2) 
    node[pos=\arrowlabelpos, left, \othercolor] {$x_{2}^{0}$};
\draw[->, >=latex, line width=0.75pt, \othercolor] (L3C) -- (L4O3) 
    node[pos=\arrowlabelpos, left, \othercolor] {$x_{3}^{0}$};
\node[right=3pt, \othercolor] at ($(L4O3) + (10pt, 0)$) {$\dotsc$};

\node[below=8pt] at (L4C) {$U\left(x_{*}^{0}, \dotsc, x_{*}^{N}\right)$};

\end{tikzpicture}
  }
  \caption{Recursive Inspection as an imperfect recall game; nodes are labelled by the information available for making the decision at that node. Note how the decision tree is in the reverse order of the inspection order: $x^N$ is decided first, and $x^0$ last.}
  \label{fig:imperfect}
\end{figure}
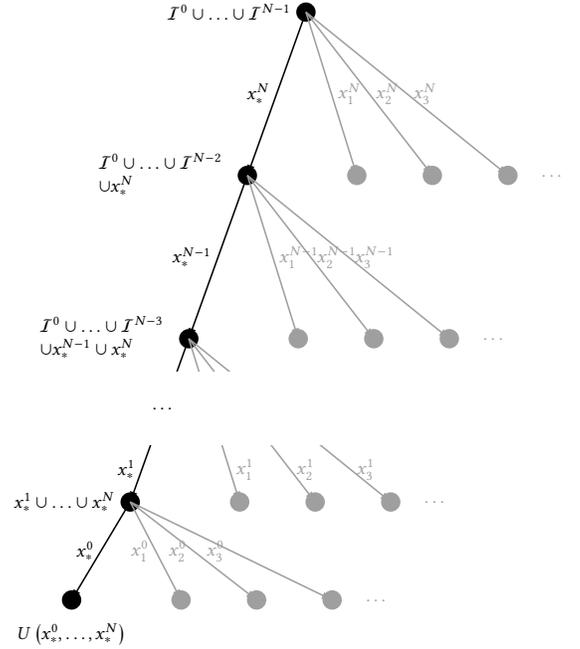

In what sense is this algorithm ``optimal''?\todo{phrase this better} It is instructive to first consider what notions of optimality \emph{don't} hold for the recursive inspection protocol:

\begin{nontheorem}
We might think that the resulting sequence $(x^0_*,\dots x^N_*)$ is optimal given all the information present in the system $\infooffers^0,\dots\infooffers^{N-1}$, i.e.
\begin{equation*}
  \mathbf{x}_* = \argmax_{\mathbf{x}\in\prod\Actions}\Expectof{
    U(\mathbf{x})\mid \infooffers^0\cup\dots\cup\infooffers^{N-1}
  }
\end{equation*}
\end{nontheorem}
\begin{proof}[Counter-example]
It is always best to simply buy $x^0_*$ and none of the subsequent $x^n_*$s that facilitated this decision. The same argument applies at any level, thus we cannot even make any claim like ``$\mathbf{x}_*$ is optimal given the information purchased''.
\end{proof}

Instead, our algorithm seems to be optimal in a ``bounded rationality'' sense\todo{bounded rationality ref}: optimal in a way that also accounts for the \emph{costs} of acquiring the information that would help improve our decision. One way to phrase this is: ex-ante (not knowing the information it is about to be offered), an agent would prefer to use this protocol compared to any other protocol.

\begin{definition}[Admissible purchase protocol]
  an ``admissible purchase protocol'' is a list of functions $\xi^n$ mapping a decision problem $\Actions^0$ and a sequence of information offer sets $\infooffers^0,\dots\infooffers^{N-1}$ generated from it, to $\Actions,\Powerset{\infooffers^0},\dots\Powerset{\infooffers^{N-1}}$:
  \begin{align*}
    x^N &= \xi^N(\infooffers^0,\dots\infooffers^{N-1}) \\
    \dots&\\
    x^n &= \xi^n(\infooffers^0,\dots\infooffers^{n-1}, x^{n+1},\dots x^N) \\
    \dots&\\
    x^0 &= \xi^0(x^1,\dots, x^N) \\
  \end{align*}
  i.e. a decision cannot ``steal'' information offers specifically made to help improve that decision, it can only depend on \emph{purchased} information offers. We denote the tupled function as $(x^0,\dots x^N)=\xi(\infooffers^0,\dots\infooffers^{N-1})=\xi(\infooffers)$.
  \label{def:app}
\end{definition}

In order to express ``ex-ante expected utility'', the agent must have a prior on the $\infooffers^n$ that will be generated --- this can be achieved either by having a measurable map $\Samples\to\finset(\InfoOffers)^N$ (where $\InfoOffers=\Sigalg\times\bools\times\Reals$ is the type of information goods) or more simply by assuming the random variables $I^n_k$ revealed are fixed and only having a prior over their values. Both allow us to take an expectation over $\infooffers^0,\dots\infooffers^{N-1}$.

\begin{theorem}[Recursive Inspections are ex-ante superior to any admissible protocol]
  Let $x^n_*$ be the protocol described in \Cref{eq:imperfect-x,eq:imperfect-u,eq:imperfect-u0}. Then for any admissible purchase protocol $\xi$:
  \begin{multline*}
    \Expectofu[\infooffers^0,\dots\infooffers^{N-1}]{U(\xi(\infooffers))} \\
    \le \Expectofu[\infooffers^0,\dots\infooffers^{N-1}]{U(x_*(\infooffers))}
  \end{multline*}
  \label{thm:exa}
\end{theorem}
\begin{proof}
Let $\xi_{*n}$ denote the following admissible protocol:
  \begin{align*}
    x^N &= \xi^N(\infooffers^0,\dots\infooffers^{N-1}) \\
    \dots&\\
    x^{n+1} &= \xi^{N+1}(\infooffers^0,\dots\infooffers^{n},x^{n+2},\dots x^N)\\
    x^n &= x_*^n(\infooffers^0,\dots\infooffers^{n-1}, x^{n+1},\dots x^N) \\
    \dots&\\
    x^0 &= x_*^0(x^1,\dots, x^N) \\
  \end{align*}
  Then $\xi_{*N}=x_*$ and $\xi_{*-1}=\xi$. It suffices to show that $\xi_{*n+1}$ dominates $\xi_{*n}$, i.e.
  \begin{multline*}
    \Expectofu[\infooffers^0,\dots\infooffers^{N-1}]{U(\xi_{*n}(\infooffers))} \\
    \le \Expectofu[\infooffers^0,\dots\infooffers^{N-1}]{U(\xi_{*n+1}(\infooffers))}
  \end{multline*}
  Observe that (where we have omitted the parameters to $x^n_*$, $\xi^n$ etc. as a shorthand):
  \begin{align*}
  U(\xi_{*n}(\infooffers^0,\dots\infooffers^{N-1})) &= U\left(
    x^0_*, \dots x^n_*,\xi^{n+1},\dots\xi^N
  \right) \\
  &= U\left(\xi^{n+1},\dots\xi^N\right)
  \end{align*}
  Where we have recursively applied \Cref{eq:imperfect-u} to transform the expression into one about the value of a node at level $n+1$. Thus the goal is reduced to:
  \begin{multline*}
    \Expectofu[\infooffers^0,\dots\infooffers^{N-1}]{U(\xi^{n+1},\dots \xi^N)} \\
    \le \Expectofu[\infooffers^0,\dots\infooffers^{N-1}]{U(x_*^{n+1}, \xi^{n+1},\dots \xi^N)}
  \end{multline*}
  From \Cref{eq:imperfect-x} we have that this is true for the expectation over $\infooffers^0,\dots\infooffers^{n}$; we can then simply take the expectation over $\infooffers^{n+1},\dots\infooffers^{N-1}$ and have our result.
\end{proof}

\section{Human feedback for Scalable Oversight}
\label{sec:scoring}

\NewDocumentCommand{\allinfo}{}{\mathbf{K}}

The mechanism in \cref{app:nonnaive} is unsuitable for settings where the information provided by the sellers must be expensively generated (rather than cheaply retrieving known information) in response to the query---in particular, this is the case with training AI models. Instead, in this setting we may assume that we can initiate as many instances as we want of the AI model we are trying to align: $\beta^1,\beta^2,\dots$ with identical information $\allinfo=\langle K, k, \infty\rangle$. Then let them recursively generate information:

\begin{itemize}
\item $\beta^1$ generates $x^1$ to help us decide our original problem $x^0$.
\item $\beta^2$ generates $x^2$, which could affect our decision on the original problem either directly or by influencing our evaluation of $x^1$
\item $\beta^3$ generates $x^3$, which could affect our decision on the original problem either directly or by influencing our evaluation of $x^1,x^2$
\item $\dots$ until some $\beta^N$ estimates that any $x^N$ it generates will only get less reward than giving $\zero$
\end{itemize}

When the mechanism terminates, the human evaluator calculates the rewards $R^n$ for each $x^n$ taking into account the full sequence of information $(x^1,x^2,\dots)$ received (the computation of these rewards will depend on the exact mechanism). In all this we regard the actions $x^n$, as before, to be purely a tuple of a random variable, its value and its price, i.e. an $I\le K$ so that its value $i$ is a function of the value $k$ of $K$ and the prices of combined random variables are additive.

It is important to let $N$ go as high as needed; for any fixed number of agents they may collude to obtain the greatest possible total reward since it's not a zero-sum game. The possibility of another agent coming in and invalidating that, should prevent collusion.

Intuitively, the idea is that if $x^1$ is bad, i.e. $\mathbf{E}[U^1(x^1)\mid x^1]$ is high but $\mathbf{E}[U^1(x^1)\mid K]$ is low, then $\beta^2$ can easily generate an $x^2$ from $K$ such that $\mathbf{E}[U^1(x^1)|x^1,x^2]$ is low. And we would reward $x^2$ for this, because it has significantly impacted our evaluation of $x^1$ and---in our view, now that we know $x^1,x^2$---in a good way. Of course this $x^2$ could actually be bad---i.e. there could be some $x^3$ that makes us revise down our estimate of $U^2(x^2)$, i.e. which tells us that $R^1(x^1\mid x^1,x^2)$ actually worsened/wasn't a great improvement over $R^1(x^1\mid x^1)$.

The following gives an example of such a mechanism, which may be seen as a generalization of AI safety via market-making to tasks beyond binary forecasting.

\begin{definition}[Marginal value mechanism]
At each point after $x^1,\dots x^n$ is generated, we note down what our action would be, if given just this much information: 

$$x^0_n=\argmax_{x^0}\mathbf{E}[U(x^0)|x^1,\dots x^n]$$

Then the ``true'' value of each successive piece of information is

$$U^n(x^n)=U(x^0_n)-U(x^0_{n-1}) - p(x^n)$$

We do not know these true values; however at the end of the process we can estimate these values based on all the information we have received.

$$R^n=\mathbf{E}[U^n(x^n)|x^0,\dots x^N]$$

And supply that as reward to $\beta^n$.
\label{def:mvm}
\end{definition}

\begin{definition}[Equilibrium]
Let $\sigma:(x^1,\dots x^{n-1})\mapsto x^n$ denote strategies and let $H(x^1,\dots x^n,\sigma^{n+1},\sigma^{n+2},\dots)$ denote the terminal history resulting from applying strategies $\sigma^{n+1},\sigma^{n+2},\dots$ starting from a pre-set history. Then $(\sigma^1_*,\sigma^2_*,\dots)$ is a \emph{subgame-perfect equilibrium} of the described game if for all $n$ and any pre-set history $h^{n-1}=(x^1,\dots x^{n-1})$:

$$\sigma^n_*(h^{n-1}) = \arg\max_{x^n}R^n(H(h^{n-1},x^n,\sigma^{n+1}_*,\sigma^{n+2}_*,\dots))$$

The actual played moves are then: $x^1_*=\sigma^1_*(\cdot)$ and $x^n_*=\sigma^n_*(x^1_*,\dots x^{n-1}_*)$.
\label{def:eq}
\end{definition}

In order to characterize the equilibrium of the marginal value mechanism game, we take inspiration from AI safety via debate \citeps{irvingAISafetyDebate2018}, where the provider of the first argument is incentivized to produce an ``irrefutable'' argument $x^1$, i.e. one such that $\forall x^2,\exists x^3, \dots\operatorname{human}(x^1,\dots x^N)=1$ in favour of $x^1$. Similarly in our setting, we call a piece of information ``inextensible'' if no future player has a profitable inextensible move. Formally:

\begin{definition}[Inextensibility]
Information $y$ ``extends'' $x^n$ (denoted $y/x^n$) if $\mathbf{E}[U^{n+1}(y)|x^n,y]\ge 0$ and call information $x^1$ inextensible (denoted $[x]$) if:

$$\forall x^2/x^1, \exists x^3/(x^1,x^2), [x^1,x^2,x^3]$$
\end{definition}

Thus $x^1$ is inextensible if

\begin{itemize}
\item $\not\exists x^2/x^1$ OR
\item $\forall x^2/x^1,\exists x^3/(x^1,x^2),\not\exists x^4/(x^1,x^2,x^3)$ OR
\item $\forall x^2/x^1,\exists x^3/x^{1,2},\forall x^4/x^{1\dots 3},\exists x^5/x^{1\dots 4},\not\exists x^6/x^{1\dots 5}$ OR
\item \dots
\end{itemize}

\begin{theorem}[Characterization of equilibrium]
At the subgame-perfect equilibrium of the marginal value mechanism:
\begin{itemize}
\item $x^1_*$ is inextensible.
\item $\forall n>1,x^n_*=\zero$
\item Among all inextensible $x^1$, $x^1_*$ has the highest ex-post VOI: $\mathbf{E}[U^1(x^1_*)\mid x^1_*]\ge\mathbf{E}[U^1(x^1)\mid x^1]$.
\end{itemize}
\label{thm:eq}
\end{theorem}

\begin{proof}
We argue by backward induction on subgames.

Fix any history $h^{n-1}=(x^1,\dots,x^{n-1})$. By \Cref{def:eq}, player $n$ chooses
\[
x^n_*=\arg\max_{x^n}R^n(H(h^{n-1},x^n,\sigma_*^{n+1},\sigma_*^{n+2},\dots)).
\]
Under \Cref{def:mvm}, the null action $\zero$ yields zero marginal contribution and zero price, hence payoff $0$ in that subgame. Therefore a non-null move is chosen at stage $n$ only if it has nonnegative continuation value relative to $\zero$.

Now consider the continuation game after some $x^1$. The predicate $y/x^n$ in the inextensibility definition is exactly the condition that player $n\!+\!1$ has a weakly profitable extension. Thus:
\begin{itemize}
\item if $\not\exists x^{n+1}/(x^1,\dots,x^n)$, then in that subgame player $n\!+\!1$'s best response is $\zero$;
\item if such an extension exists but can be countered at the next step, then (by subgame perfection) that extension is not part of an optimal continuation unless the counter-counter-continuation is itself unprofitable.
\end{itemize}
Hence the alternating quantifiers in \Cref{def:eq} and in the definition of inextensibility coincide: $[x^1]$ means player $2$ has no profitable continuation in equilibrium from $x^1$.

Therefore, in any SPE, $x^1_*$ must be inextensible; otherwise player $2$ would have a profitable deviation in the subgame after $x^1_*$, contradicting subgame perfection. This proves the first bullet.

Given $[x^1_*]$, player $2$'s equilibrium action is $\zero$. Once $x^2_*=\zero$, the same argument applies recursively to every later subgame, so for all $n>1$ we get $x^n_*=\zero$. This proves the second bullet.

With continuation fixed at zeros, player $1$'s payoff from choosing $x^1$ is exactly its own ex-post marginal value:
\[
R^1=\mathbf{E}[U^1(x^1)\mid x^1].
\]
Hence player $1$ solves
\[
\max_{x^1:[x^1]}\mathbf{E}[U^1(x^1)\mid x^1],
\]
so the equilibrium choice $x^1_*$ is an inextensible $x^1$ with maximal ex-post VOI. This is the third bullet.
\end{proof}

\section{Practical algorithm}
\label{sec:practical}

\Cref{alg:im} shows the inspection protocol for information markets introduced in \citets{weissRedesigningInformationMarkets2024}: information sellers offer information goods to a buyer, who spins off an LLM to inspect and purchase the information goods. Our \emph{Recursive Inspection Protocol} extends this by letting the subcontracted LLM buyer further consult the information market (spin off another sub-LLM) to help it make its decision, ad recursum. A simplified version of the logic, ignoring server implementation details, is presented in in \Cref{alg:rim} and \Cref{fig:rim}.

\begin{algorithm}[tb]
\caption{One-level Inspection Protocol from \citets{weissRedesigningInformationMarkets2024}}
\label{alg:im}
\begin{algorithmic}

\Struct{\BuyerContext}
  \State $\Actions$ : decision problem it wants information for
  \State $D$ : $\operatorname{list}[\operatorname{str}]\to\Actions$, decision procedure based on available information
\EndStruct

\Struct{\Seller}
  \State $A$ : $\BuyerContext\to\operatorname{str}\times\Reals$, generate \InfoOffer{} and price
\EndStruct

\Procedure{IP}{$Q:\BuyerContext$}
\State\LeftComment{Post contextual information to sellers to receive \InfoOffers{}}
\State $\infooffers\gets \{\beta(Q)\text{ for }\beta\in\Sellers\}$
\State\LeftComment{Use an LLM to decide which $\infogood$ to buy}
\State $I^*\gets\LLM(\prompt=\text{``You need to buy an $\infogood$ from $\infooffers$ to help decide $Q$''})()$
\State\LeftComment{Decide based on purchased information}
\State $x^*\gets Q.D(I^*)$
\State \textbf{return} $x^*$
\EndProcedure

\end{algorithmic}
\end{algorithm}

\begin{algorithm}[tb]
\caption{Recursive Inspection Protocol}
\label{alg:rim}
\begin{algorithmic}

\Procedure{RIP}{$Q:\BuyerContext$}
\State\LeftComment{Post to sellers and get \InfoOffers{} from them}
\State $\infooffers\gets \{\beta(Q)\text{ for }\beta\in\Sellers\}$
\State\LeftComment{Create Recursive BuyerContext to help decide $Q$}
\State $Q' \gets \BuyerContext(\infooffers, \LLM(\prompt=\text{``You need to buy an $\infogood$ from $\infooffers$ to help decide $Q$''}))$
\State \LeftComment{Get \InfoOffers{} chosen in $Q'$}
\State $I^*\gets\RIP(Q')$
\State\LeftComment{Decide based on all collected information from recursive steps}
\State $x^*\gets Q.D(I^*)$
\State \textbf{return} $x^*,I^*$
\EndProcedure

\end{algorithmic}
\end{algorithm}

\begin{figure}
  \centering
  \resizebox{\columnwidth}{!}{
  \input{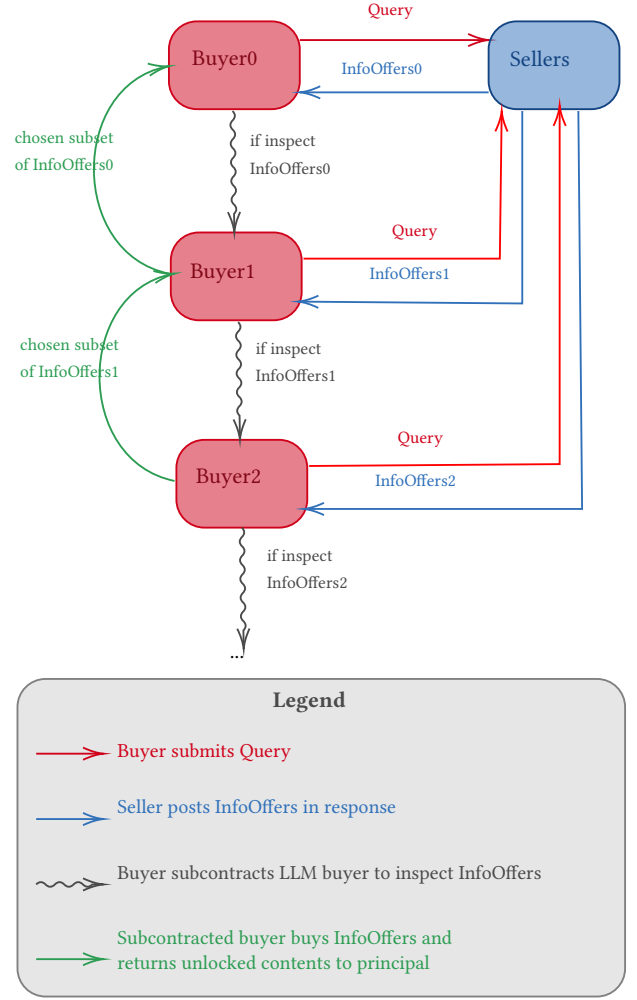}
  }
  \caption{The Recursive Inspection Protocol}
  \label{fig:rim}
\end{figure}

\begin{figure}[htbp]
  \centering
  \begin{subfigure}[b]{\columnwidth}
    \centering
    \includegraphics[width=0.9\linewidth]{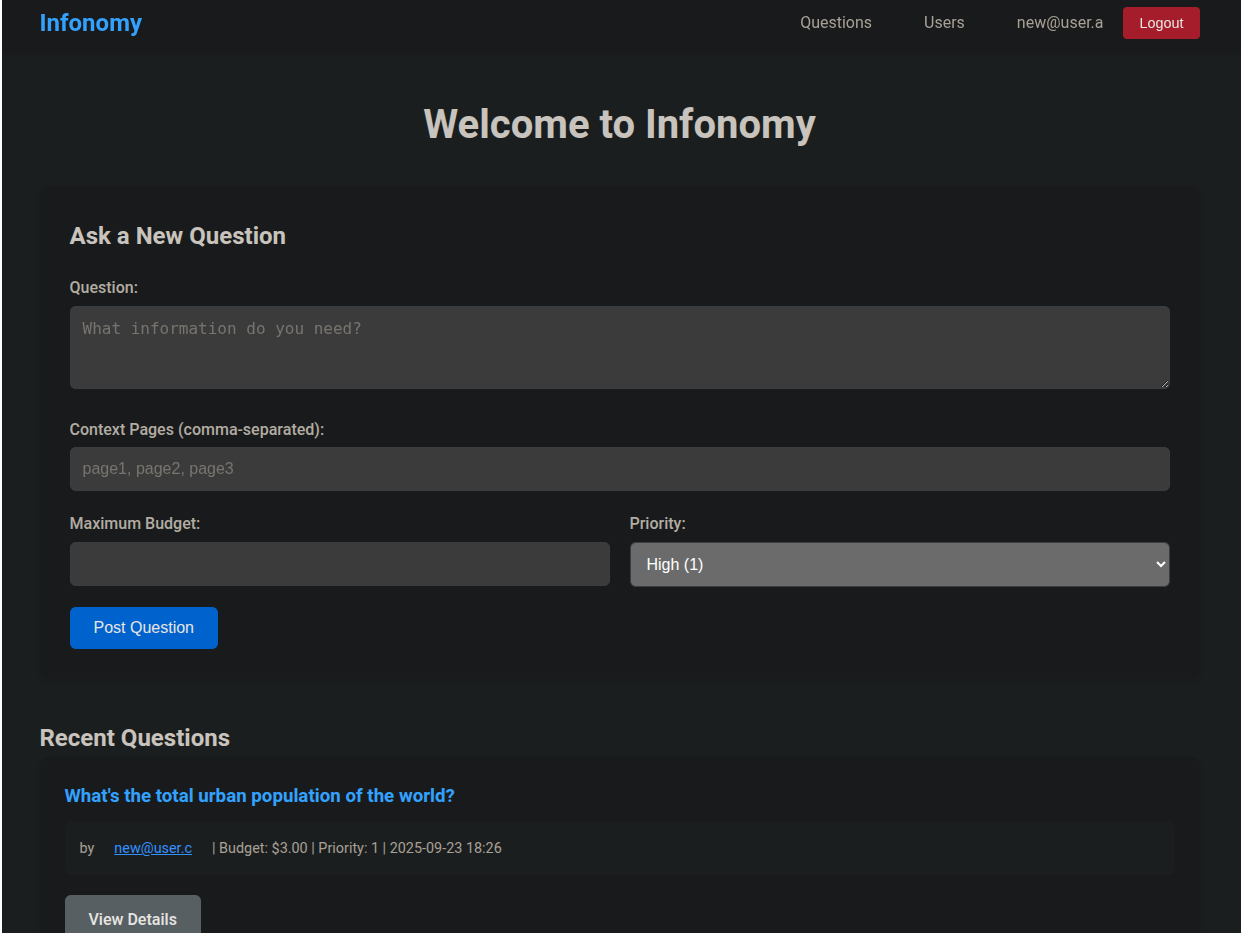}
    \caption{Posting a new question (\BuyerContext); viewing \BuyerContexts on the server}
    \label{fig:infonomy-1}
  \end{subfigure}

  \begin{subfigure}[b]{\columnwidth}
    \centering
    \includegraphics[width=0.9\linewidth]{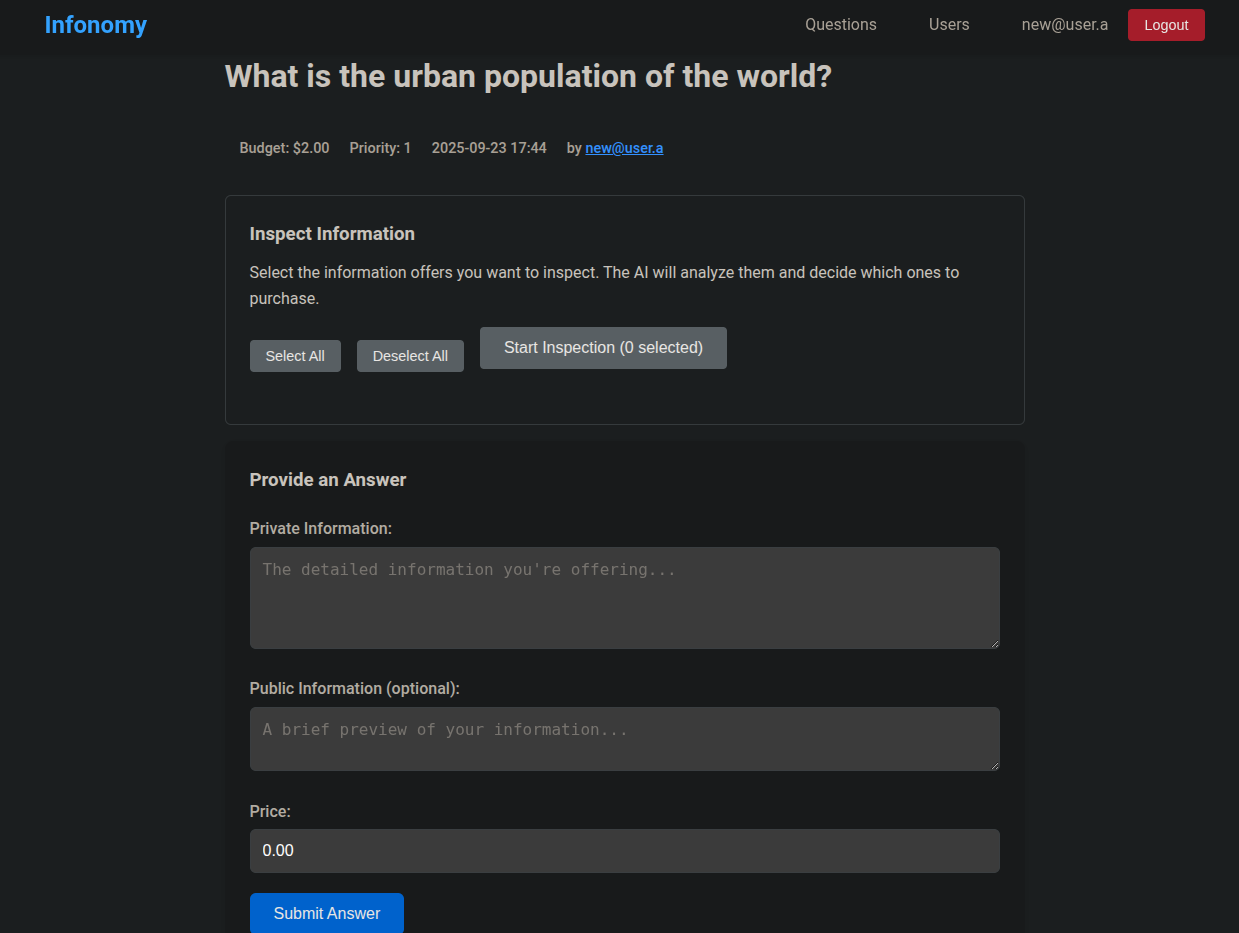}
    \caption{Posting an answer (\InfoOffer{}); initiating a recursive inspection}
    \label{fig:infonomy-2}
  \end{subfigure}

  \begin{subfigure}[b]{\columnwidth}
    \centering
    \includegraphics[width=0.9\linewidth]{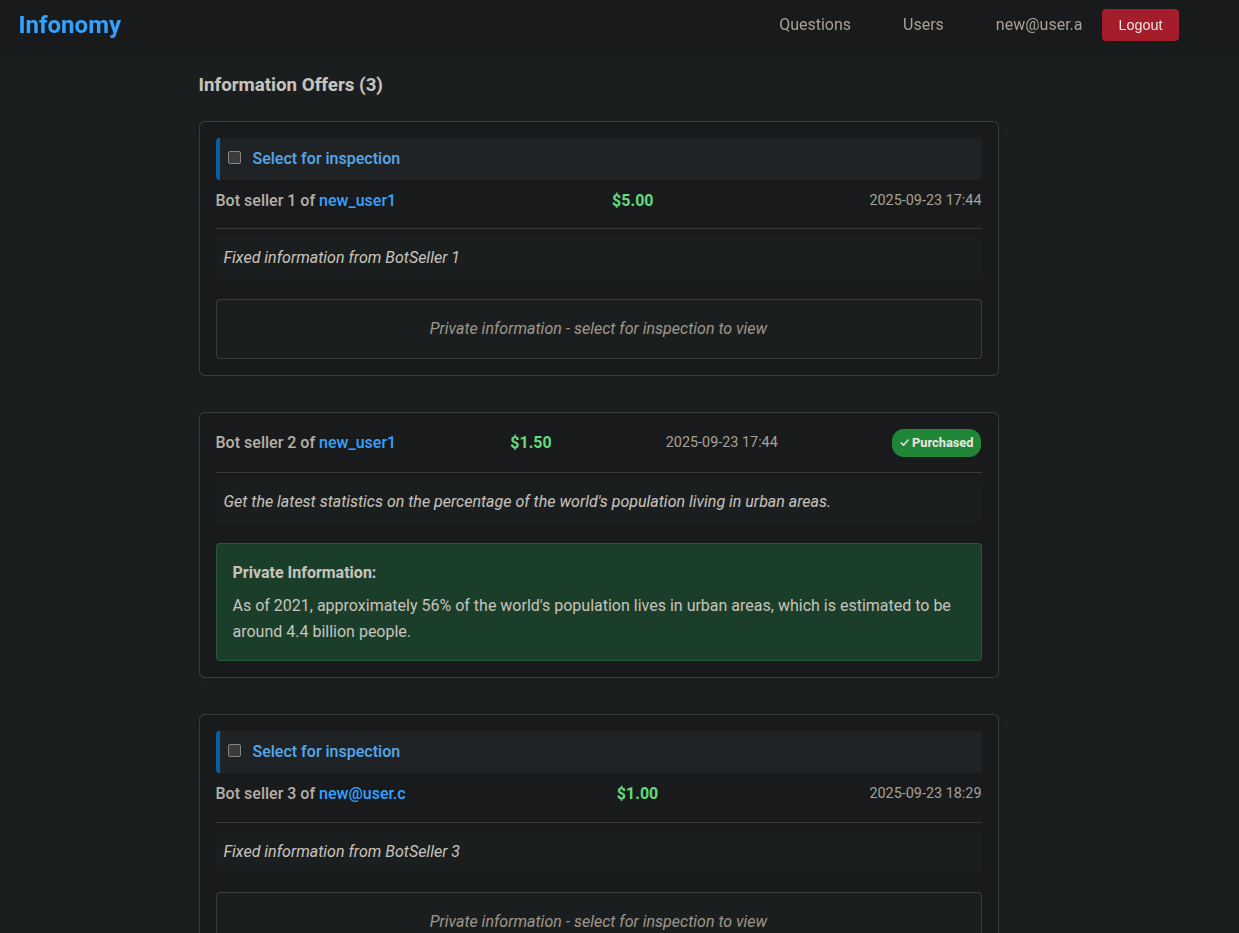}
    \caption{Inspecting and purchasing \InfoOffers{}; viewing purchased \InfoOffers{}}
    \label{fig:infonomy-3}
  \end{subfigure}
\end{figure}

\begin{figure}[htbp]
  \ContinuedFloat
  \centering
  \begin{subfigure}[b]{\columnwidth}
    \centering
    \includegraphics[width=0.9\linewidth]{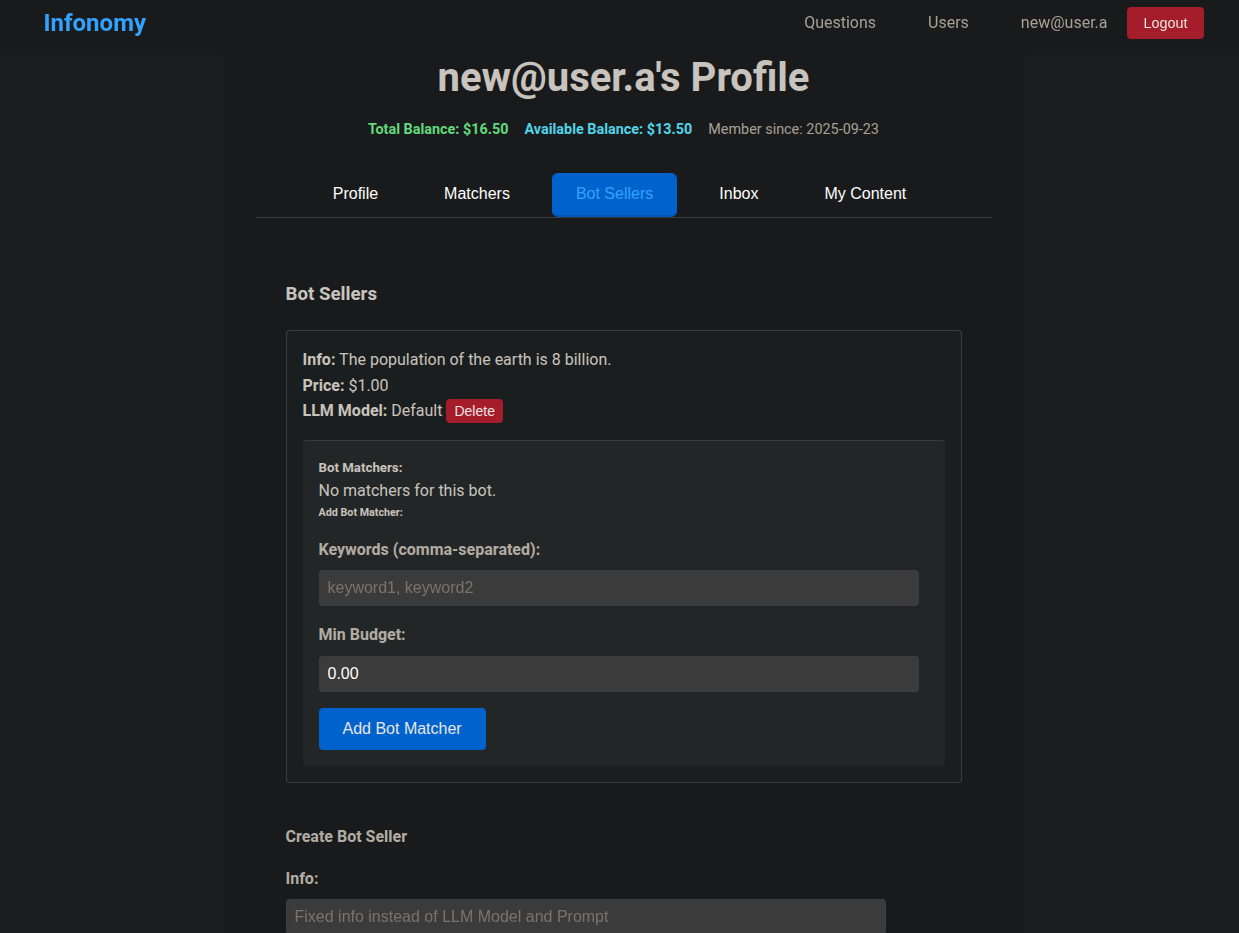}
    \caption{Bot sellers automatically answer recursive \BuyerContexts}
    \label{fig:infonomy-4}
  \end{subfigure}

  \caption{Screenshots from the \texttt{infonomy-server} platform}
  \label{fig:infonomy}
\end{figure}

A working implementation of an information market server implementing the Recursive Inspection Protocol is available at the \infonomy{} repository\footnote{\infonomyurl} -- some screenshots of the platform's GUI are shown in \Cref{fig:infonomy}. Many applications of such a server are immediate:

\begin{itemize}
\item \emph{Question \& Answer site} -- As presented, \infonomy{} can be seen as a Q \& A site with market incentives for answering questions.
\item \emph{Privatized product regulation} -- The \BuyerContexts{} might be names or links to products (the decision being ``should I buy?''), and the \InfoOffers{} could be inspection checks done by private labs, or customer reviews, which are now incentivized to be answerable to the customer.
\item \emph{Community Notes} -- In the spirit of well-known crowdsourced fact-checking systems such as Community Notes/Birdwatch \citeps{birdwatch2022}, we might use information markets as a ``comments section for the internet'', where \BuyerContexts{} would be links to webpages or social media posts (the decision being ``should I believe?'') and the \InfoOffers{} would be fact-checks or important context.
\item \emph{Reasoning in prediction markets} -- Forecasters on prediction markets may benefit from incentivizing the provision of information relevant to a forecast. One solution to this was given by \citets{conitzerPredictionMarketsMechanism2012}; another is to use \infonomy{}: where the \BuyerContexts{} are the questions being forecasted on (``What is the correct probability of this question?'') and the \InfoOffers{} are any relevant pieces of information.
\end{itemize}

\section{Future work}
\label{sec:so}

We have introduced a Bayesian framework for ``extrapolated volition'': i.e. to model a subjective buyer/rater's ``most fully informed'' score for the value of some piece of information. This allows us to design mechanisms for information markets (\cref{app:nonnaive}) as well as for supplying human feedback to AI models (\cref{sec:scoring}). 

The ideal desideratum we would like a scalable oversight mechanism to satisfy is that it should incentivize the AI/information-seller to give the optimal information according to the information \emph{it} possesses: if the seller has information $K=(I_1,I_2,\dots)$, it should give information $I$ that optimizes\footnote{Naively one may think this means the AI should simply give $K$---realistically though, $K$ would be very large, reflecting the AI's entire knowledge and capabilities. Thus taking the cost of $K$ into account (assumed to be $\infty$ in \cref{sec:scoring}), that would certainly not be optimal.} $\mathbf{E}[U^1(I)|K=k]$. This would exactly describe the buyer's extrapolated volition: \emph{what would the buyer do if she were as smart as the AI?} Or we could say: this would exactly \emph{align} the AI or information-seller to our own values, while maintaining its superior information. 

Unfortunately, our marginal value mechanism does not satisfy this hope. Take the following example.

\begin{example}
Suppose our decision problem is $\{0,1\}$ and:
\begin{itemize}
\item in our prior judgement, 0 is the better choice: $\mathbf{E}[U(0)]=1$, $\mathbf{E}[U(1)]=0$
\item  $I_1$ tells us 1 is better: $\mathbf{E}[U(0)|I_1]=0$, $\mathbf{E}[U(1)|I_1]=1$
\item  $I_2$ refutes $I_1$ and says 0 is better: $\mathbf{E}[U(0)|I_1,I_2]=1$ and $\mathbf{E}[U(1)|I_1,I_2]=0$
\item  $I_3$ refutes $I_2$ and says 1 is better: $\mathbf{E}[U(0)|I_1,I_2,I_3]=0$ and $\mathbf{E}[U(1)|I_1,I_2,I_3]=1$
\item with the full information, 1 is the better choice: $\mathbf{E}[U(0)|K]=0$, $\mathbf{E}[U(1)|K]=1$
\end{itemize}

But say $I_1$ and $I_2$ are cheap, while $p(I_3)=100$. Then the best information to reveal is $I_1$ --- but it won't be revealed, because $I_2$ will cheaply refute it while defending it with $I_3$ is too expensive. Thus instead we want to say that the agent can't give information \emph{so bad} that its shortfall exceeds its ``cost of defense'' (in this case, 100).

\end{example}

So while it is \emph{not} generally true that $x^1_*=\arg\max_{x^1}\mathbf{E}[U^1(x^1)|K]$, we may hope for a lower bound on ``how bad'' the equilibrium could possibly get, i.e. a result like $\mathbf{E}[U^1(x^1_*)|K]\ge \max_{x^1}\mathbf{E}[U^1(x^1)|K]-\mathcal{E}$ for some shortfall expression $\mathcal{E}$ that is some measure of the ``cost of defending the correct information''---and we may also use the expression for such a shortfall as a measure of how good a particular scalable oversight protocol is.

\bibliographystyle{ACM-Reference-Format}
\bibliography{refs}


\begin{thebibliography}{43}


\ifx \showCODEN    \undefined \def \showCODEN     #1{\unskip}     \fi
\ifx \showDOI      \undefined \def \showDOI       #1{#1}\fi
\ifx \showISBNx    \undefined \def \showISBNx     #1{\unskip}     \fi
\ifx \showISBNxiii \undefined \def \showISBNxiii  #1{\unskip}     \fi
\ifx \showISSN     \undefined \def \showISSN      #1{\unskip}     \fi
\ifx \showLCCN     \undefined \def \showLCCN      #1{\unskip}     \fi
\ifx \shownote     \undefined \def \shownote      #1{#1}          \fi
\ifx \showarticletitle \undefined \def \showarticletitle #1{#1}   \fi
\ifx \showURL      \undefined \def \showURL       {\relax}        \fi
\providecommand\bibfield[2]{#2}
\providecommand\bibinfo[2]{#2}
\providecommand\natexlab[1]{#1}
\providecommand\showeprint[2][]{arXiv:#2}

\bibitem[\protect\citeauthoryear{Akerlof}{Akerlof}{1978}]%
        {lemons}
\bibfield{author}{\bibinfo{person}{George~A Akerlof}.} \bibinfo{year}{1978}\natexlab{}.
\newblock \showarticletitle{The market for “lemons”: Quality uncertainty and the market mechanism}.
\newblock In \bibinfo{booktitle}{\emph{Uncertainty in economics}}. \bibinfo{publisher}{Elsevier}, \bibinfo{pages}{235--251}.
\newblock


\bibitem[\protect\citeauthoryear{Arrow}{Arrow}{1972}]%
        {arrow1972economic}
\bibfield{author}{\bibinfo{person}{K.~J. Arrow}.} \bibinfo{year}{1972}\natexlab{}.
\newblock \bibinfo{booktitle}{\emph{Economic Welfare and the Allocation of Resources for Invention}}.
\newblock \bibinfo{publisher}{Macmillan Education UK}, \bibinfo{address}{London}, \bibinfo{pages}{219--236}.
\newblock
\showISBNx{978-1-349-15486-9}
\urldef\tempurl%
\url{https://doi.org/10.1007/978-1-349-15486-9_13}
\showDOI{\tempurl}


\bibitem[\protect\citeauthoryear{Babaioff, Kleinberg, and Paes~Leme}{Babaioff et~al\mbox{.}}{2012}]%
        {babiaoff2012OptimalMechanismsSellingInformation}
\bibfield{author}{\bibinfo{person}{Moshe Babaioff}, \bibinfo{person}{Robert Kleinberg}, {and} \bibinfo{person}{Renato Paes~Leme}.} \bibinfo{year}{2012}\natexlab{}.
\newblock \showarticletitle{Optimal mechanisms for selling information}. In \bibinfo{booktitle}{\emph{Proceedings of the 13th ACM Conference on Electronic Commerce}} (Valencia, Spain) \emph{(\bibinfo{series}{EC '12})}. \bibinfo{publisher}{Association for Computing Machinery}, \bibinfo{address}{New York, NY, USA}, \bibinfo{pages}{92–109}.
\newblock
\showISBNx{9781450314152}
\urldef\tempurl%
\url{https://doi.org/10.1145/2229012.2229024}
\showDOI{\tempurl}


\bibitem[\protect\citeauthoryear{Ben-Or, Goldreich, Goldwasser, H{\aa}stad, Kilian, Micali, and Rogaway}{Ben-Or et~al\mbox{.}}{1990}]%
        {zkPSPACE2}
\bibfield{author}{\bibinfo{person}{Michael Ben-Or}, \bibinfo{person}{Oded Goldreich}, \bibinfo{person}{Shafi Goldwasser}, \bibinfo{person}{Johan H{\aa}stad}, \bibinfo{person}{Joe Kilian}, \bibinfo{person}{Silvio Micali}, {and} \bibinfo{person}{Phillip Rogaway}.} \bibinfo{year}{1990}\natexlab{}.
\newblock \showarticletitle{Everything provable is provable in zero-knowledge}. In \bibinfo{booktitle}{\emph{Advances in Cryptology – CRYPTO 1988 - Proceedings}} \emph{(\bibinfo{series}{Lecture Notes in Computer Science (including subseries Lecture Notes in Artificial Intelligence and Lecture Notes in Bioinformatics)})}, \bibfield{editor}{\bibinfo{person}{Shafi Goldwasser}} (Ed.). \bibinfo{publisher}{Springer Verlag}, \bibinfo{address}{Germany}, \bibinfo{pages}{37--56}.
\newblock
\showISBNx{9780387971964}
\urldef\tempurl%
\url{https://doi.org/10.1007/0-387-34799-2_4}
\showDOI{\tempurl}
\newblock
\shownote{Publisher Copyright: {\textcopyright} Springer-Verlag Berlin Heidelberg 1990.; Conference on Theory and Applications of Cryptography, CRYPTO 1988 ; Conference date: 21-08-1988 Through 25-08-1988.}


\bibitem[\protect\citeauthoryear{Bergemann, Bonatti, and Smolin}{Bergemann et~al\mbox{.}}{2018}]%
        {bergemannDesignPriceInformation2018}
\bibfield{author}{\bibinfo{person}{Dirk Bergemann}, \bibinfo{person}{Alessandro Bonatti}, {and} \bibinfo{person}{Alex Smolin}.} \bibinfo{year}{2018}\natexlab{}.
\newblock \showarticletitle{The {{Design}} and {{Price}} of {{Information}}}.
\newblock \bibinfo{journal}{\emph{The American Economic Review}} \bibinfo{volume}{108}, \bibinfo{number}{1} (\bibinfo{year}{2018}), \bibinfo{pages}{1--48}.
\newblock
\showISSN{0002-8282}
\showeprint[jstor]{26527944}


\bibitem[\protect\citeauthoryear{Bowman, Hyun, Perez, Chen, Pettit, Heiner, Luko{\v s}i{\=u}t{\.e}, Askell, Jones, Chen, Goldie, Mirhoseini, McKinnon, Olah, Amodei, Amodei, Drain, Li, {Tran-Johnson}, Kernion, Kerr, Mueller, Ladish, Landau, Ndousse, Lovitt, Elhage, Schiefer, Joseph, Mercado, DasSarma, Larson, McCandlish, Kundu, Johnston, Kravec, Showk, Fort, {Telleen-Lawton}, Brown, Henighan, Hume, Bai, {Hatfield-Dodds}, Mann, and Kaplan}{Bowman et~al\mbox{.}}{2022}]%
        {bowmanMeasuringProgressScalable2022}
\bibfield{author}{\bibinfo{person}{Samuel~R. Bowman}, \bibinfo{person}{Jeeyoon Hyun}, \bibinfo{person}{Ethan Perez}, \bibinfo{person}{Edwin Chen}, \bibinfo{person}{Craig Pettit}, \bibinfo{person}{Scott Heiner}, \bibinfo{person}{Kamil{\.e} Luko{\v s}i{\=u}t{\.e}}, \bibinfo{person}{Amanda Askell}, \bibinfo{person}{Andy Jones}, \bibinfo{person}{Anna Chen}, \bibinfo{person}{Anna Goldie}, \bibinfo{person}{Azalia Mirhoseini}, \bibinfo{person}{Cameron McKinnon}, \bibinfo{person}{Christopher Olah}, \bibinfo{person}{Daniela Amodei}, \bibinfo{person}{Dario Amodei}, \bibinfo{person}{Dawn Drain}, \bibinfo{person}{Dustin Li}, \bibinfo{person}{Eli {Tran-Johnson}}, \bibinfo{person}{Jackson Kernion}, \bibinfo{person}{Jamie Kerr}, \bibinfo{person}{Jared Mueller}, \bibinfo{person}{Jeffrey Ladish}, \bibinfo{person}{Joshua Landau}, \bibinfo{person}{Kamal Ndousse}, \bibinfo{person}{Liane Lovitt}, \bibinfo{person}{Nelson Elhage}, \bibinfo{person}{Nicholas Schiefer}, \bibinfo{person}{Nicholas Joseph}, \bibinfo{person}{Noem{\'i} Mercado}, \bibinfo{person}{Nova DasSarma}, \bibinfo{person}{Robin Larson}, \bibinfo{person}{Sam McCandlish}, \bibinfo{person}{Sandipan Kundu}, \bibinfo{person}{Scott Johnston}, \bibinfo{person}{Shauna Kravec}, \bibinfo{person}{Sheer~El Showk}, \bibinfo{person}{Stanislav Fort}, \bibinfo{person}{Timothy {Telleen-Lawton}}, \bibinfo{person}{Tom Brown}, \bibinfo{person}{Tom Henighan}, \bibinfo{person}{Tristan Hume}, \bibinfo{person}{Yuntao Bai}, \bibinfo{person}{Zac {Hatfield-Dodds}}, \bibinfo{person}{Ben Mann}, {and} \bibinfo{person}{Jared Kaplan}.} \bibinfo{year}{2022}\natexlab{}.
\newblock \bibinfo{title}{Measuring {{Progress}} on {{Scalable Oversight}} for {{Large Language Models}}}.
\newblock
\newblock
\urldef\tempurl%
\url{https://doi.org/10.48550/arXiv.2211.03540}
\showDOI{\tempurl}
\showeprint[arxiv]{2211.03540}


\bibitem[\protect\citeauthoryear{Burns, Izmailov, Kirchner, Baker, Gao, Aschenbrenner, Chen, Ecoffet, Joglekar, Leike, Sutskever, and Wu}{Burns et~al\mbox{.}}{2024}]%
        {burnsWeaktoStrongGeneralizationEliciting2024}
\bibfield{author}{\bibinfo{person}{Collin Burns}, \bibinfo{person}{Pavel Izmailov}, \bibinfo{person}{Jan~Hendrik Kirchner}, \bibinfo{person}{Bowen Baker}, \bibinfo{person}{Leo Gao}, \bibinfo{person}{Leopold Aschenbrenner}, \bibinfo{person}{Yining Chen}, \bibinfo{person}{Adrien Ecoffet}, \bibinfo{person}{Manas Joglekar}, \bibinfo{person}{Jan Leike}, \bibinfo{person}{Ilya Sutskever}, {and} \bibinfo{person}{Jeffrey Wu}.} \bibinfo{year}{2024}\natexlab{}.
\newblock \showarticletitle{Weak-to-{{Strong Generalization}}: {{Eliciting Strong Capabilities With Weak Supervision}}}. In \bibinfo{booktitle}{\emph{Proceedings of the 41st {{International Conference}} on {{Machine Learning}}}}. \bibinfo{publisher}{PMLR}, \bibinfo{pages}{4971--5012}.
\newblock
\showISSN{2640-3498}


\bibitem[\protect\citeauthoryear{Buterin}{Buterin}{2024}]%
        {buterin2024InfoFinance}
\bibfield{author}{\bibinfo{person}{Vitalik Buterin}.} \bibinfo{year}{2024}\natexlab{}.
\newblock \bibinfo{title}{From prediction markets to info finance}.
\newblock \bibinfo{howpublished}{\url{https://vitalik.eth.limo/general/2024/11/09/infofinance.html}}.
\newblock
\urldef\tempurl%
\url{https://vitalik.eth.limo/general/2024/11/09/infofinance.html}
\showURL{%
\tempurl}
\newblock
\shownote{Accessed: \today.}


\bibitem[\protect\citeauthoryear{Casper, Davies, Shi, Gilbert, Scheurer, Rando, Freedman, Korbak, Lindner, Freire, Wang, Marks, Segerie, Carroll, Peng, Christoffersen, Damani, Slocum, Anwar, Siththaranjan, Nadeau, Michaud, Pfau, Krasheninnikov, Chen, Langosco, Hase, Biyik, Dragan, Krueger, Sadigh, and Hadfield-Menell}{Casper et~al\mbox{.}}{2023}]%
        {caspar2023openrlhf}
\bibfield{author}{\bibinfo{person}{Stephen Casper}, \bibinfo{person}{Xander Davies}, \bibinfo{person}{Claudia Shi}, \bibinfo{person}{Thomas~Krendl Gilbert}, \bibinfo{person}{J{\'e}r{\'e}my Scheurer}, \bibinfo{person}{Javier Rando}, \bibinfo{person}{Rachel Freedman}, \bibinfo{person}{Tomek Korbak}, \bibinfo{person}{David Lindner}, \bibinfo{person}{Pedro Freire}, \bibinfo{person}{Tony~Tong Wang}, \bibinfo{person}{Samuel Marks}, \bibinfo{person}{Charbel-Raphael Segerie}, \bibinfo{person}{Micah Carroll}, \bibinfo{person}{Andi Peng}, \bibinfo{person}{Phillip~J.K. Christoffersen}, \bibinfo{person}{Mehul Damani}, \bibinfo{person}{Stewart Slocum}, \bibinfo{person}{Usman Anwar}, \bibinfo{person}{Anand Siththaranjan}, \bibinfo{person}{Max Nadeau}, \bibinfo{person}{Eric~J Michaud}, \bibinfo{person}{Jacob Pfau}, \bibinfo{person}{Dmitrii Krasheninnikov}, \bibinfo{person}{Xin Chen}, \bibinfo{person}{Lauro Langosco}, \bibinfo{person}{Peter Hase}, \bibinfo{person}{Erdem Biyik}, \bibinfo{person}{Anca Dragan}, \bibinfo{person}{David Krueger}, \bibinfo{person}{Dorsa Sadigh}, {and} \bibinfo{person}{Dylan Hadfield-Menell}.} \bibinfo{year}{2023}\natexlab{}.
\newblock \showarticletitle{Open Problems and Fundamental Limitations of Reinforcement Learning from Human Feedback}.
\newblock \bibinfo{journal}{\emph{Transactions on Machine Learning Research}} (\bibinfo{year}{2023}).
\newblock
\showISSN{2835-8856}
\urldef\tempurl%
\url{https://openreview.net/forum?id=bx24KpJ4Eb}
\showURL{%
\tempurl}
\newblock
\shownote{Survey Certification, Featured Certification.}


\bibitem[\protect\citeauthoryear{Chen, Li, and Xu}{Chen et~al\mbox{.}}{2022}]%
        {chenSellingDataMachine2022}
\bibfield{author}{\bibinfo{person}{Junjie Chen}, \bibinfo{person}{Minming Li}, {and} \bibinfo{person}{Haifeng Xu}.} \bibinfo{year}{2022}\natexlab{}.
\newblock \showarticletitle{Selling {{Data To}} a {{Machine Learner}}: {{Pricing}} via {{Costly Signaling}}}. In \bibinfo{booktitle}{\emph{Proceedings of the 39th {{International Conference}} on {{Machine Learning}}}}. \bibinfo{publisher}{PMLR}, \bibinfo{pages}{3336--3359}.
\newblock
\showISSN{2640-3498}


\bibitem[\protect\citeauthoryear{Christiano, Shlegeris, and Amodei}{Christiano et~al\mbox{.}}{2018}]%
        {christiano2018IDA}
\bibfield{author}{\bibinfo{person}{Paul Christiano}, \bibinfo{person}{Buck Shlegeris}, {and} \bibinfo{person}{Dario Amodei}.} \bibinfo{year}{2018}\natexlab{}.
\newblock \bibinfo{title}{Supervising strong learners by amplifying weak experts}.
\newblock
\newblock
\showeprint[arxiv]{1810.08575}~[cs.LG]
\urldef\tempurl%
\url{https://arxiv.org/abs/1810.08575}
\showURL{%
\tempurl}


\bibitem[\protect\citeauthoryear{Conitzer}{Conitzer}{2009}]%
        {conitzerPredictionMarketsMechanism2012}
\bibfield{author}{\bibinfo{person}{Vincent Conitzer}.} \bibinfo{year}{2009}\natexlab{}.
\newblock \showarticletitle{Prediction markets, mechanism design, and cooperative game theory}. In \bibinfo{booktitle}{\emph{Proceedings of the Twenty-Fifth Conference on Uncertainty in Artificial Intelligence}} (Montreal, Quebec, Canada) \emph{(\bibinfo{series}{UAI '09})}. \bibinfo{publisher}{AUAI Press}, \bibinfo{address}{Arlington, Virginia, USA}, \bibinfo{pages}{101–108}.
\newblock
\showISBNx{9780974903958}


\bibitem[\protect\citeauthoryear{D{\"u}tting, Mirrokni, Paes~Leme, Xu, and Zuo}{D{\"u}tting et~al\mbox{.}}{2024}]%
        {duttingMechanismDesignLarge2024}
\bibfield{author}{\bibinfo{person}{Paul D{\"u}tting}, \bibinfo{person}{Vahab Mirrokni}, \bibinfo{person}{Renato Paes~Leme}, \bibinfo{person}{Haifeng Xu}, {and} \bibinfo{person}{Song Zuo}.} \bibinfo{year}{2024}\natexlab{}.
\newblock \showarticletitle{Mechanism {{Design}} for {{Large Language Models}}}. In \bibinfo{booktitle}{\emph{Proceedings of the {{ACM}} on {{Web Conference}} 2024}} \emph{(\bibinfo{series}{{{WWW}} '24})}. \bibinfo{publisher}{Association for Computing Machinery}, \bibinfo{address}{New York, NY, USA}, \bibinfo{pages}{144--155}.
\newblock
\showISBNx{9798400701719}
\urldef\tempurl%
\url{https://doi.org/10.1145/3589334.3645511}
\showDOI{\tempurl}


\bibitem[\protect\citeauthoryear{Fallah, Jordan, Makhdoumi, and Malekian}{Fallah et~al\mbox{.}}{2024}]%
        {fallah2024ThreeLayerDataMarkets}
\bibfield{author}{\bibinfo{person}{Alireza Fallah}, \bibinfo{person}{Michael Jordan}, \bibinfo{person}{Ali Makhdoumi}, {and} \bibinfo{person}{Azarakhsh Malekian}.} \bibinfo{year}{2024}\natexlab{}.
\newblock \showarticletitle{On Three-Layer Data Markets}.
\newblock \bibinfo{journal}{\emph{ArXiv}}  \bibinfo{volume}{abs/2402.09697} (\bibinfo{year}{2024}).
\newblock
\urldef\tempurl%
\url{https://api.semanticscholar.org/CorpusID:267682401}
\showURL{%
\tempurl}


\bibitem[\protect\citeauthoryear{Fallenstein and Soares}{Fallenstein and Soares}{2015}]%
        {vingeanReflection}
\bibfield{author}{\bibinfo{person}{Benja Fallenstein} {and} \bibinfo{person}{Nate Soares}.} \bibinfo{year}{2015}\natexlab{}.
\newblock \bibinfo{booktitle}{\emph{Vingean Reflection: Reliable Reasoning for Self-Improving Agents}}.
\newblock \bibinfo{type}{Technical Report} 2015-2. \bibinfo{institution}{MIRI}.
\newblock
\urldef\tempurl%
\url{https://intelligence.org/files/VingeanReflection.pdf}
\showURL{%
\tempurl}


\bibitem[\protect\citeauthoryear{Ghorbani and Zou}{Ghorbani and Zou}{2019}]%
        {ghorbani2019DataShapley}
\bibfield{author}{\bibinfo{person}{Amirata Ghorbani} {and} \bibinfo{person}{James Zou}.} \bibinfo{year}{2019}\natexlab{}.
\newblock \showarticletitle{Data Shapley: Equitable Valuation of Data for Machine Learning}. In \bibinfo{booktitle}{\emph{Proceedings of the 36th International Conference on Machine Learning}} \emph{(\bibinfo{series}{Proceedings of Machine Learning Research}, Vol.~\bibinfo{volume}{97})}, \bibfield{editor}{\bibinfo{person}{Kamalika Chaudhuri} {and} \bibinfo{person}{Ruslan Salakhutdinov}} (Eds.). \bibinfo{publisher}{PMLR}, \bibinfo{pages}{2242--2251}.
\newblock
\urldef\tempurl%
\url{https://proceedings.mlr.press/v97/ghorbani19c.html}
\showURL{%
\tempurl}


\bibitem[\protect\citeauthoryear{Goldreich}{Goldreich}{2001}]%
        {goldreich2001FoundationsCryptographyV1}
\bibfield{author}{\bibinfo{person}{O. Goldreich}.} \bibinfo{year}{2001}\natexlab{}.
\newblock \bibinfo{booktitle}{\emph{Foundations of Cryptography: Volume 1, Basic Tools}}.
\newblock \bibinfo{publisher}{Cambridge University Press}.
\newblock
\showISBNx{9780521791724}
\urldef\tempurl%
\url{https://books.google.co.uk/books?id=oo3RzgEACAAJ}
\showURL{%
\tempurl}


\bibitem[\protect\citeauthoryear{Halawi, Zhang, {Yueh-Han}, and Steinhardt}{Halawi et~al\mbox{.}}{2024}]%
        {halawiApproachingHumanLevelForecasting2024}
\bibfield{author}{\bibinfo{person}{Danny Halawi}, \bibinfo{person}{Fred Zhang}, \bibinfo{person}{Chen {Yueh-Han}}, {and} \bibinfo{person}{Jacob Steinhardt}.} \bibinfo{year}{2024}\natexlab{}.
\newblock \bibinfo{title}{Approaching {{Human-Level Forecasting}} with {{Language Models}}}.
\newblock
\newblock
\urldef\tempurl%
\url{https://doi.org/10.48550/arXiv.2402.18563}
\showDOI{\tempurl}
\showeprint[arxiv]{2402.18563}~[cs]


\bibitem[\protect\citeauthoryear{Hammond and {Adam-Day}}{Hammond and {Adam-Day}}{2024}]%
        {hammondNeuralInteractiveProofs2024}
\bibfield{author}{\bibinfo{person}{Lewis Hammond} {and} \bibinfo{person}{Sam {Adam-Day}}.} \bibinfo{year}{2024}\natexlab{}.
\newblock \showarticletitle{Neural {{Interactive Proofs}}}. In \bibinfo{booktitle}{\emph{{{ICML}} 2024 {{Next Generation}} of {{AI Safety Workshop}}}}.
\newblock


\bibitem[\protect\citeauthoryear{Hanson}{Hanson}{2002}]%
        {hansonLogarithmicMarketScoring2002}
\bibfield{author}{\bibinfo{person}{Robin Hanson}.} \bibinfo{year}{2002}\natexlab{}.
\newblock \showarticletitle{Logarithmic {{Market Scoring Rules}} for {{Modular Combinatorial Information Aggregation}}}.
\newblock \bibinfo{journal}{\emph{The Journal of Prediction Markets}} \bibinfo{volume}{1}, \bibinfo{number}{1} (\bibinfo{date}{January} \bibinfo{year}{2002}), \bibinfo{pages}{3--15}.
\newblock
\urldef\tempurl%
\url{https://doi.org/10.5750/jpm.v1i1.417}
\showDOI{\tempurl}


\bibitem[\protect\citeauthoryear{Hanson}{Hanson}{2011a}]%
        {hansonIPBarbedWire2011}
\bibfield{author}{\bibinfo{person}{Robin Hanson}.} \bibinfo{year}{2011}\natexlab{a}.
\newblock \bibinfo{title}{{{IP}}+ {{Like Barbed Wire}}?}
\newblock
\newblock


\bibitem[\protect\citeauthoryear{Hanson}{Hanson}{2011b}]%
        {hansonRahEfficientIP2011}
\bibfield{author}{\bibinfo{person}{Robin Hanson}.} \bibinfo{year}{2011}\natexlab{b}.
\newblock \bibinfo{title}{Rah {{Efficient IP}}}.
\newblock
\newblock


\bibitem[\protect\citeauthoryear{Howard}{Howard}{1966}]%
        {howard1966information}
\bibfield{author}{\bibinfo{person}{R. Howard}.} \bibinfo{year}{1966}\natexlab{}.
\newblock \showarticletitle{Information Value Theory}.
\newblock \bibinfo{journal}{\emph{IEEE Transactions on Systems Science and Cybernetics}} \bibinfo{volume}{2}, \bibinfo{number}{1} (\bibinfo{year}{1966}), \bibinfo{pages}{22--26}.
\newblock
\urldef\tempurl%
\url{https://doi.org/10.1109/tssc.1966.30007}
\showDOI{\tempurl}


\bibitem[\protect\citeauthoryear{Hubinger}{Hubinger}{2020a}]%
        {evanhubingerAISafetyMarket2020}
\bibfield{author}{\bibinfo{person}{Evan Hubinger}.} \bibinfo{year}{2020}\natexlab{a}.
\newblock \bibinfo{title}{{{AI}} Safety via Market Making --- {{LessWrong}}}.
\newblock
\newblock


\bibitem[\protect\citeauthoryear{Hubinger}{Hubinger}{2020b}]%
        {hubinger2020alignmentProposalsComplexityClasses}
\bibfield{author}{\bibinfo{person}{Evan Hubinger}.} \bibinfo{year}{2020}\natexlab{b}.
\newblock \bibinfo{title}{Alignment proposals and complexity classes}.
\newblock \bibinfo{howpublished}{\url{https://www.lesswrong.com/posts/N64THGX7XNCqRtvPG/alignment-proposals-and-complexity-classes}}.
\newblock
\urldef\tempurl%
\url{https://www.lesswrong.com/posts/N64THGX7XNCqRtvPG/alignment-proposals-and-complexity-classes}
\showURL{%
Retrieved \today from \tempurl}
\newblock
\shownote{Accessed: \today.}


\bibitem[\protect\citeauthoryear{Impagliazzo and Yung}{Impagliazzo and Yung}{1987}]%
        {zkPSPACE1}
\bibfield{author}{\bibinfo{person}{Russell Impagliazzo} {and} \bibinfo{person}{Moti Yung}.} \bibinfo{year}{1987}\natexlab{}.
\newblock \showarticletitle{Direct Minimum-Knowledge Computations}. In \bibinfo{booktitle}{\emph{A Conference on the Theory and Applications of Cryptographic Techniques on Advances in Cryptology}} \emph{(\bibinfo{series}{CRYPTO '87})}. \bibinfo{publisher}{Springer-Verlag}, \bibinfo{address}{Berlin, Heidelberg}, \bibinfo{pages}{40–51}.
\newblock
\showISBNx{3540187960}


\bibitem[\protect\citeauthoryear{Irving, Christiano, and Amodei}{Irving et~al\mbox{.}}{2018}]%
        {irvingAISafetyDebate2018}
\bibfield{author}{\bibinfo{person}{Geoffrey Irving}, \bibinfo{person}{Paul Christiano}, {and} \bibinfo{person}{Dario Amodei}.} \bibinfo{year}{2018}\natexlab{}.
\newblock \bibinfo{title}{{{AI}} Safety via Debate}.
\newblock
\newblock
\urldef\tempurl%
\url{https://doi.org/10.48550/arXiv.1805.00899}
\showDOI{\tempurl}
\showeprint[arxiv]{1805.00899}~[cs, stat]


\bibitem[\protect\citeauthoryear{Kuhn, Arrow, Barankin, Blackwell, Bott, Dalkey, Dresher, Gale, Gillies, Glicksberg, Gross, Karlin, Kuhn, Mayberry, Milnor, Motzkin, von Neumann, Raiffa, Shapley, Shiffman, Stewart, Thompson, and Thrall}{Kuhn et~al\mbox{.}}{1953}]%
        {kuhn1953information}
\bibfield{author}{\bibinfo{person}{H.~W. Kuhn}, \bibinfo{person}{K.~J. Arrow}, \bibinfo{person}{E.~W. Barankin}, \bibinfo{person}{D. Blackwell}, \bibinfo{person}{R. Bott}, \bibinfo{person}{N. Dalkey}, \bibinfo{person}{M. Dresher}, \bibinfo{person}{D. Gale}, \bibinfo{person}{D.~B. Gillies}, \bibinfo{person}{I. Glicksberg}, \bibinfo{person}{O. Gross}, \bibinfo{person}{S. Karlin}, \bibinfo{person}{H.~W. Kuhn}, \bibinfo{person}{J.~P. Mayberry}, \bibinfo{person}{J.~W. Milnor}, \bibinfo{person}{T.~S. Motzkin}, \bibinfo{person}{J. von Neumann}, \bibinfo{person}{H. Raiffa}, \bibinfo{person}{L.~S. Shapley}, \bibinfo{person}{M. Shiffman}, \bibinfo{person}{F.~M. Stewart}, \bibinfo{person}{G.~L. Thompson}, {and} \bibinfo{person}{R.~M. Thrall}.} \bibinfo{year}{1953}\natexlab{}.
\newblock \bibinfo{booktitle}{\emph{Extensive games and the problem of information}}.
\newblock \bibinfo{publisher}{Princeton University Press}, \bibinfo{pages}{193--216}.
\newblock
\showISBNx{9780691079356}
\urldef\tempurl%
\url{http://www.jstor.org/stable/j.ctt1b9x1zv.17}
\showURL{%
\tempurl}


\bibitem[\protect\citeauthoryear{Li, Gao, Li, Li, and Liao}{Li et~al\mbox{.}}{2024}]%
        {liEconAgentLargeLanguage2024a}
\bibfield{author}{\bibinfo{person}{Nian Li}, \bibinfo{person}{Chen Gao}, \bibinfo{person}{Mingyu Li}, \bibinfo{person}{Yong Li}, {and} \bibinfo{person}{Qingmin Liao}.} \bibinfo{year}{2024}\natexlab{}.
\newblock \showarticletitle{{{EconAgent}}: {{Large Language Model-Empowered Agents}} for {{Simulating Macroeconomic Activities}}}. In \bibinfo{booktitle}{\emph{Proceedings of the 62nd {{Annual Meeting}} of the {{Association}} for {{Computational Linguistics}} ({{Volume}} 1: {{Long Papers}})}}, \bibfield{editor}{\bibinfo{person}{Lun-Wei Ku}, \bibinfo{person}{Andre Martins}, {and} \bibinfo{person}{Vivek Srikumar}} (Eds.). \bibinfo{publisher}{Association for Computational Linguistics}, \bibinfo{address}{Bangkok, Thailand}, \bibinfo{pages}{15523--15536}.
\newblock
\urldef\tempurl%
\url{https://doi.org/10.18653/v1/2024.acl-long.829}
\showDOI{\tempurl}


\bibitem[\protect\citeauthoryear{Lindley}{Lindley}{1956}]%
        {lindley1956OnMeasureInformationExperiment}
\bibfield{author}{\bibinfo{person}{D.~V. Lindley}.} \bibinfo{year}{1956}\natexlab{}.
\newblock \showarticletitle{On a Measure of the Information Provided by an Experiment}.
\newblock \bibinfo{journal}{\emph{The Annals of Mathematical Statistics}} \bibinfo{volume}{27}, \bibinfo{number}{4} (\bibinfo{year}{1956}), \bibinfo{pages}{986--1005}.
\newblock
\showISSN{00034851, 21688990}
\urldef\tempurl%
\url{http://www.jstor.org/stable/2237191}
\showURL{%
\tempurl}


\bibitem[\protect\citeauthoryear{Manelli and Vincent}{Manelli and Vincent}{2006}]%
        {manelli2006Bundling}
\bibfield{author}{\bibinfo{person}{Alejandro~M. Manelli} {and} \bibinfo{person}{Daniel~R. Vincent}.} \bibinfo{year}{2006}\natexlab{}.
\newblock \showarticletitle{Bundling as an optimal selling mechanism for a multiple-good monopolist}.
\newblock \bibinfo{journal}{\emph{Journal of Economic Theory}} \bibinfo{volume}{127}, \bibinfo{number}{1} (\bibinfo{year}{2006}), \bibinfo{pages}{1--35}.
\newblock
\showISSN{0022-0531}
\urldef\tempurl%
\url{https://doi.org/10.1016/j.jet.2005.08.007}
\showDOI{\tempurl}


\bibitem[\protect\citeauthoryear{Paleka, Sudhir, Alvarez, Bhat, Shen, Wang, and Tram{\`e}r}{Paleka et~al\mbox{.}}{2024}]%
        {palekaConsistencyChecksLanguage2024}
\bibfield{author}{\bibinfo{person}{Daniel Paleka}, \bibinfo{person}{Abhimanyu~Pallavi Sudhir}, \bibinfo{person}{Alejandro Alvarez}, \bibinfo{person}{Vineeth Bhat}, \bibinfo{person}{Adam Shen}, \bibinfo{person}{Evan Wang}, {and} \bibinfo{person}{Florian Tram{\`e}r}.} \bibinfo{year}{2024}\natexlab{}.
\newblock \showarticletitle{Consistency {{Checks}} for {{Language Model Forecasters}}}. In \bibinfo{booktitle}{\emph{The {{Thirteenth International Conference}} on {{Learning Representations}}}}.
\newblock


\bibitem[\protect\citeauthoryear{Raiffa and Schlaifer}{Raiffa and Schlaifer}{1961}]%
        {raiffa1961AppliedStatisticalDecisionTheory}
\bibfield{author}{\bibinfo{person}{H. Raiffa} {and} \bibinfo{person}{R. Schlaifer}.} \bibinfo{year}{1961}\natexlab{}.
\newblock \bibinfo{booktitle}{\emph{Applied Statistical Decision Theory}}.
\newblock \bibinfo{publisher}{Division of Research, Graduate School of Business Adminitration, Harvard University}.
\newblock
\showISBNx{9780875840178}
\showLCCN{60011282}
\urldef\tempurl%
\url{https://books.google.co.uk/books?id=wPBLAAAAMAAJ}
\showURL{%
\tempurl}


\bibitem[\protect\citeauthoryear{Samuelson and Nordhaus}{Samuelson and Nordhaus}{2009}]%
        {samuelson2009economics}
\bibfield{author}{\bibinfo{person}{P.A. Samuelson} {and} \bibinfo{person}{W.D. Nordhaus}.} \bibinfo{year}{2009}\natexlab{}.
\newblock \bibinfo{booktitle}{\emph{Economics}}.
\newblock \bibinfo{publisher}{McGraw-Hill Education}.
\newblock
\showISBNx{9780073511290}
\showLCCN{2009003178}
\urldef\tempurl%
\url{https://books.google.co.uk/books?id=eS5ZAAAAYAAJ}
\showURL{%
\tempurl}


\bibitem[\protect\citeauthoryear{Schoenegger, Tuminauskaite, Park, and Tetlock}{Schoenegger et~al\mbox{.}}{2024}]%
        {schoeneggerWisdomSiliconCrowd2024}
\bibfield{author}{\bibinfo{person}{Philipp Schoenegger}, \bibinfo{person}{Indre Tuminauskaite}, \bibinfo{person}{Peter~S. Park}, {and} \bibinfo{person}{Philip~E. Tetlock}.} \bibinfo{year}{2024}\natexlab{}.
\newblock \bibinfo{title}{Wisdom of the {{Silicon Crowd}}: {{LLM Ensemble Prediction Capabilities Rival Human Crowd Accuracy}}}.
\newblock
\newblock
\urldef\tempurl%
\url{https://doi.org/10.48550/arXiv.2402.19379}
\showDOI{\tempurl}
\showeprint[arxiv]{2402.19379}~[cs]


\bibitem[\protect\citeauthoryear{Stigler}{Stigler}{1961}]%
        {stigler1961economics}
\bibfield{author}{\bibinfo{person}{George~J Stigler}.} \bibinfo{year}{1961}\natexlab{}.
\newblock \showarticletitle{The economics of information}.
\newblock \bibinfo{journal}{\emph{Journal of political economy}} \bibinfo{volume}{69}, \bibinfo{number}{3} (\bibinfo{year}{1961}), \bibinfo{pages}{213--225}.
\newblock


\bibitem[\protect\citeauthoryear{Sudhir, Kaunismaa, and Panickssery}{Sudhir et~al\mbox{.}}{2025}]%
        {sudhir2025ScalableOversightBenchmark}
\bibfield{author}{\bibinfo{person}{Abhimanyu~Pallavi Sudhir}, \bibinfo{person}{Jackson Kaunismaa}, {and} \bibinfo{person}{Arjun Panickssery}.} \bibinfo{year}{2025}\natexlab{}.
\newblock \showarticletitle{A Benchmark for Scalable Oversight Mechanisms}. In \bibinfo{booktitle}{\emph{ICLR 2025 Workshop on Bidirectional Human-AI Alignment}}.
\newblock
\urldef\tempurl%
\url{https://openreview.net/forum?id=mzLBxX84VI}
\showURL{%
\tempurl}


\bibitem[\protect\citeauthoryear{Tewolde, Zhang, Oesterheld, Zampetakis, Sandholm, Goldberg, and Conitzer}{Tewolde et~al\mbox{.}}{2024}]%
        {imperfectRecall2024}
\bibfield{author}{\bibinfo{person}{Emanuel Tewolde}, \bibinfo{person}{Brian~Hu Zhang}, \bibinfo{person}{Caspar Oesterheld}, \bibinfo{person}{Manolis Zampetakis}, \bibinfo{person}{Tuomas Sandholm}, \bibinfo{person}{Paul Goldberg}, {and} \bibinfo{person}{Vincent Conitzer}.} \bibinfo{year}{2024}\natexlab{}.
\newblock \showarticletitle{Imperfect-recall games: equilibrium concepts and their complexity}. In \bibinfo{booktitle}{\emph{Proceedings of the Thirty-Third International Joint Conference on Artificial Intelligence}} (Jeju, Korea) \emph{(\bibinfo{series}{IJCAI '24})}. Article \bibinfo{articleno}{332}, \bibinfo{numpages}{11}~pages.
\newblock
\showISBNx{978-1-956792-04-1}
\urldef\tempurl%
\url{https://doi.org/10.24963/ijcai.2024/332}
\showDOI{\tempurl}


\bibitem[\protect\citeauthoryear{Thomassen, Vassbø, Solheim-Kile, and Lohne}{Thomassen et~al\mbox{.}}{2016}]%
        {tenders}
\bibfield{author}{\bibinfo{person}{Kristine Thomassen}, \bibinfo{person}{Siril Vassbø}, \bibinfo{person}{Espen Solheim-Kile}, {and} \bibinfo{person}{Jardar Lohne}.} \bibinfo{year}{2016}\natexlab{}.
\newblock \showarticletitle{Public-Private Partnership: Transaction Costs of Tendering}.
\newblock \bibinfo{journal}{\emph{Procedia Computer Science}}  \bibinfo{volume}{100} (\bibinfo{year}{2016}), \bibinfo{pages}{818--825}.
\newblock
\showISSN{1877-0509}
\urldef\tempurl%
\url{https://doi.org/10.1016/j.procs.2016.09.230}
\showDOI{\tempurl}
\newblock
\shownote{International Conference on ENTERprise Information Systems/International Conference on Project MANagement/International Conference on Health and Social Care Information Systems and Technologies, CENTERIS/ProjMAN / HCist 2016.}


\bibitem[\protect\citeauthoryear{Van~Alstyne}{Van~Alstyne}{1999a}]%
        {vanAlstyne1999AproposalValuingInformationInstrumentalGoods}
\bibfield{author}{\bibinfo{person}{Marshall~V. Van~Alstyne}.} \bibinfo{year}{1999}\natexlab{a}.
\newblock \showarticletitle{A proposal for valuing information and instrumental goods}. In \bibinfo{booktitle}{\emph{Proceedings of the 20th International Conference on Information Systems}} (Charlotte, North Carolina, USA) \emph{(\bibinfo{series}{ICIS '99})}. \bibinfo{publisher}{Association for Information Systems}, \bibinfo{address}{USA}, \bibinfo{pages}{328–345}.
\newblock


\bibitem[\protect\citeauthoryear{Van~Alstyne}{Van~Alstyne}{1999b}]%
        {vanalstyneProposalValuingInformation1999}
\bibfield{author}{\bibinfo{person}{Marshall~V. Van~Alstyne}.} \bibinfo{year}{1999}\natexlab{b}.
\newblock \showarticletitle{A Proposal for Valuing Information and Instrumental Goods}. In \bibinfo{booktitle}{\emph{Proceedings of the 20th International Conference on {{Information Systems}}}} \emph{(\bibinfo{series}{{{ICIS}} '99})}. \bibinfo{publisher}{Association for Information Systems}, \bibinfo{address}{USA}, \bibinfo{pages}{328--345}.
\newblock


\bibitem[\protect\citeauthoryear{Weiss, Rahaman, Wuthrich, Bengio, Li, Sch{\"o}lkopf, and Pal}{Weiss et~al\mbox{.}}{2024}]%
        {weissRedesigningInformationMarkets2024}
\bibfield{author}{\bibinfo{person}{Martin Weiss}, \bibinfo{person}{Nasim Rahaman}, \bibinfo{person}{Manuel Wuthrich}, \bibinfo{person}{Yoshua Bengio}, \bibinfo{person}{Li~Erran Li}, \bibinfo{person}{Bernhard Sch{\"o}lkopf}, {and} \bibinfo{person}{Christopher Pal}.} \bibinfo{year}{2024}\natexlab{}.
\newblock \showarticletitle{Redesigning {{Information Markets}} in the {{Era}} of {{Language Models}}}. In \bibinfo{booktitle}{\emph{First {{Conference}} on {{Language Modeling}}}}.
\newblock


\bibitem[\protect\citeauthoryear{Wojcik, Hilgard, Judd, Mocanu, Ragain, Hunzaker, Coleman, and Baxter}{Wojcik et~al\mbox{.}}{2022}]%
        {birdwatch2022}
\bibfield{author}{\bibinfo{person}{Stefan Wojcik}, \bibinfo{person}{Sophie Hilgard}, \bibinfo{person}{Nick Judd}, \bibinfo{person}{Delia Mocanu}, \bibinfo{person}{Stephen Ragain}, \bibinfo{person}{M.~B.~Fallin Hunzaker}, \bibinfo{person}{Keith Coleman}, {and} \bibinfo{person}{Jay Baxter}.} \bibinfo{year}{2022}\natexlab{}.
\newblock \bibinfo{title}{Birdwatch: Crowd Wisdom and Bridging Algorithms can Inform Understanding and Reduce the Spread of Misinformation}.
\newblock
\newblock
\showeprint[arxiv]{2210.15723}~[cs.SI]
\urldef\tempurl%
\url{https://arxiv.org/abs/2210.15723}
\showURL{%
\tempurl}


\end{thebibliography}

\appendix

\section{Basic results about value-of-information}
\label{sec:thms}

We include a corrected and generalized version of an incorrectly-formulated result in the arXiv version of \citet{weissRedesigningInformationMarkets2024}. These lemmas are perhaps obvious, but the whole theory of value-of-information and information bazaars rests upon them. \Cref{lem:voinf} asserts that ``a Bayesian agent expects to gain from information''\footnote{From an AI alignment perspective, this may instead be phrased ``A Bayesian agent trusts a more informed version of itself'', and is a statistical-information version of \emph{Vingean reflection} or \emph{tiling agents} \citep{vingeanReflection}, the desire that (even logically non-omniscient) agents can trust smarter versions of themselves.}. \Cref{lem:voins} asserts that ``a Bayesian agent expects to gain from inspection''. Since the agent can always choose not to buy anything after inspecting, in the absence of transaction costs of inspection buyers \emph{will} want to inspect more and more at each step. %

\begin{lemma}[Bayesian agent expects to gain from information]
Assume a probability space $(\Samples,\Sigalg,\Prob)$, a set of choices $\Actions$ and a measurable utility function $U:\Samples\times\Actions\to\Reals$. Then for any random variable $I:\Samples\to\Reals$, we have:
\begin{equation*}
\Expectofu[I]{U(\argmax_x\Expectof{U(x)\mid I})} \ge
\max\Expectof{U(x)}
\end{equation*}
\label{lem:voinf}
\end{lemma}
\begin{proof}
For all values of $I=i$ and for all $x_0\in\Actions$, we have (from the definition of $\argmax$):
\begin{multline*}
\Expectof{U(\argmax_x \Expectof{U(x)\mid I=i})\mid I=i} \\
\ge \Expectof{U(x_0)\mid I=i}
\end{multline*}
We take the expectation over $i\sim I$ on both sides, apply the law of total expectation and set $x_0=\argmax\Expectof{U(x)}$ to obtain our result.
\end{proof}

\begin{lemma}[Bayesian agent expects to gain from inspection]
Let the recursive construction of $\Actions^n$, $U^n$, $\InfoOffers^n$,$x_*^n$ be as in \Cref{sec:bayes}. Then $\forall n$, $\Expectof{U^{n+1}(x^{n+1}_*)}\ge 0$. The same also applies to finite-depth inspections i.e. $\Expectof{U^{n+1}(x^{n+1}_{*:N})}\ge 0$.
\label{lem:voins}
\end{lemma}
\begin{proof}
\Cref{eq:rim-choices} implies that $\Expectof{U^n\left(x^n_*\mid x^{n+1}_*\right)}$ is $\ge$ any version of the same expression with $x^n_*$ replaced by any $x\in\Actions^n$. We choose $x=\zero$, for which the expression evaluates to 0, and take the expectation over $x^{n+1}_*$.
\end{proof}

\end{document}